\documentclass[11pt]{article}
\usepackage{fullpage}
\usepackage{tikz}
\usetikzlibrary{arrows.meta,positioning,matrix}
\usepackage{amsmath}
\usepackage{amssymb}
\usepackage{enumitem}
\usepackage{amsthm}
\usepackage[most]{tcolorbox}
\usepackage{subfiles}

\usepackage{abraces}
\usepackage[dvipsnames]{xcolor}
\usepackage{nicematrix}
\usepackage{hyperref}
\hypersetup{
    colorlinks=true,
    linkcolor=blue,
    filecolor=magenta,      
    urlcolor=cyan,
}
\usepackage{cleveref}
\usepackage{dsfont}
\usepackage{blkarray}
\usepackage{booktabs}

\newtheorem{theorem}{Theorem}[section]
\newtheorem{corollary}[theorem]{Corollary}
\newtheorem{remark}[theorem]{Remark}
\newtheorem{lemma}[theorem]{Lemma}

\newtheorem{defi}[theorem]{Definition}

\newtheorem{observation}[theorem]{Observation}
\newtheorem{proposition}[theorem]{Proposition}

\newtheorem{fact}[theorem]{Fact}

\newtheorem{innercustomthm}{}
\newenvironment{repeattheorem}[1]
  {\renewcommand\theinnercustomthm{#1}\innercustomthm}
  {\endinnercustomthm}

\newcommand{\ket}[1]{\ensuremath{\left|#1\right\rangle}}
\newcommand{\bra}[1]{\ensuremath{\left\langle#1\right|}}

\newcommand{\norm}[1]{\ensuremath{\left\|#1\right\|}}
\newcommand{\wh}{\widehat}

\newcommand{\Q}{\ensuremath{\mathsf{Q}}}

\newcommand{\E}{\mathbb{E}}

\newcommand{\DictatorFunction}{\mathsf{Dict}}
\newcommand{\IndexFunction}{\mathsf{Index}}

\newcommand{\pmone}{\{-1,1\}}
\newcommand{\cbra}[1]{\left\{#1\right\}}
\newcommand{\rbra}[1]{\left(#1\right)}
\newcommand{\Var}{\mathbf{Var}}
\newcommand{\cQ}{\mathcal{Q}}
\newcommand{\cR}{\mathcal{R}}
\newcommand{\bR}{\mathbb{R}}
\newcommand{\bI}{\mathbb{I}}
\newcommand{\rec}{\mathsf{rec}}
\newcommand{\XOR}{\mathsf{XOR}}
\newcommand{\OR}{\mathsf{OR}}

\newcommand{\val}{\mathsf{val}}
\renewcommand{\cal}[1]{\mathcal{#1}}

\DeclareMathOperator{\ADV}{ADV}
\DeclareMathOperator*{\argmax}{arg\,max}

\allowdisplaybreaks

\title{Quantum Search With Generalized Wildcards}

\author{Arjan Cornelissen\thanks{Simons Institute for the Theory of Computing, University of California, Berkeley, United States of America {\tt ajcornelissen@outlook.com}}
\and
Nikhil S.~Mande\thanks{University of Liverpool, UK {\tt mande@liverpool.ac.uk}}
\and 
Subhasree Patro\thanks{Eindhoven University of Technology, Netherlands {\tt patrofied@gmail.com}}
\and
Nithish Raja\thanks{Eindhoven University of Technology, Netherlands {\tt n.r.raja@tue.nl}}
\and
Swagato Sanyal\thanks{University of Sheffield, UK {\tt swagato.sanyal@sheffield.ac.uk}}
}
\date{}

\begin{document}

\maketitle

\begin{abstract}
    In the ``search with wildcards'' problem [Ambainis, Montanaro, Quantum Inf.~Comput.'14], one's goal is to learn an unknown bit-string $x \in \{-1,1\}^n$. An algorithm may, at unit cost, test equality of any subset of the hidden string with a string of its choice. Ambainis and Montanaro showed a quantum algorithm of cost $O(\sqrt{n} \log n)$ and a near-matching lower bound of $\Omega(\sqrt{n})$. Belovs [Comput.~Comp.'15] subsequently showed a tight $O(\sqrt{n})$ upper bound.
    
    We consider a natural generalization of this problem, parametrized by a subset $\cal{Q} \subseteq 2^{[n]}$, where an algorithm may test whether $x_S = b$ for an arbitrary $S \in \cal{Q}$ and $b \in \{-1,1\}^S$ of its choice, at unit cost. We show the following:
    \begin{itemize}
        \item For all $k \in [n]$, when $\cal{Q}$ is the collection of all sets of size at most $k$, the quantum query complexity is $\Theta(n/\sqrt{k})$. In particular when $k = n$, this corresponds to the standard search with wildcards setting. This recovers and generalizes the tight characterization of Belovs, and Ambainis and Montanaro, using completely different techniques.
        \item When $\cal{Q}$ is the collection of contiguous blocks, the quantum query complexity is $\tilde{\Theta}(n)$.
        \item When $\cal{Q}$ is the collection of prefixes, the quantum query complexity is $\Theta(n)$.
    \end{itemize}
    All of these results are derived using a framework that we develop. We apply a symmetry reduction to the \emph{primal} version of the negative-weight adversary bound, and show that the quantum query complexity of learning $x$ is characterized, up to a constant factor, by a particular optimization program, which can be succinctly described as follows: `maximize over all odd functions $f : \{-1,1\}^n \to \mathbb{R}$ the ratio of the maximum value of $f$ to the maximum (over $T \in \cal{Q}$) standard deviation of $f$ on a subcube whose free variables are exactly $T$.'

    To the best of our knowledge, ours is the first work to use the primal version of the negative-weight adversary bound (which is a \emph{maximization} program typically used to show lower bounds) to show new quantum query \emph{upper} bounds without explicitly resorting to SDP duality. 
\end{abstract}
\thispagestyle{empty}

\newpage
\tableofcontents
\thispagestyle{empty}

\newpage
\setcounter{page}{1}
\pagenumbering{arabic}

\section{Introduction}\label{sec:intro}
Reconstructing an unknown string from limited access to its substrings is a natural and widely studied task across many areas of science and engineering. In computational biology, for example, genome assembly and DNA sequencing rely on reconstructing a long unknown strand from local substring measurements (see, for example, \cite{SS95, robertson2010novo, CWE15, LD24} and the references therein). Beyond applied motivations, the problem provides a clean and mathematically tractable model for understanding the complexity of learning a hidden string via substring queries.

Most relevant to our work is the \emph{search with wildcards} problem, introduced by Ambainis and Montanaro~\cite{AM14}. In this problem, the goal is to learn a hidden bitstring $x \in \pmone^n$ (i.e., an algorithm must correctly output $x$ with high probability) using the minimum number of queries in the following \emph{wildcard query} model. At unit cost, an algorithm may ``check its guess for $x$'' on any subset of the inputs. Formally, at unit cost, an algorithm may make a query $(S, b)$ where $S \subseteq [n], b \in \pmone^S$. 
The query returns 1 if $x_S = b$, and 0 otherwise. Wildcard queries are a generalization of the usual query model, which only allows for $|S| = 1$. Classically even in the wildcard model every query provides at most 1 bit of information, yielding an information-theoretic lower bound of $\Omega(n)$ for learning $x$. In the usual quantum query model with $|S| = 1$, even though roughly a factor-2 speedup is possible~\cite{vD98}, this is essentially the best speedup one can get~\cite{FGGS99}.

Ambainis and Montanaro~\cite{AM14} showed that in the wildcard model, quantum algorithms can learn $x$ using only $O(\sqrt{n} \log n)$ queries, exhibiting roughly a quadratic quantum speedup over the classical lower bound. However, their analysis and algorithm crucially rely on the ability to query \emph{all subsets} of coordinates.
In many practical settings, this assumption can be unrealistic. For instance, in \emph{de novo} genome sequencing and related problems in computational biology~\cite{robertson2010novo}, one has no promise on the input string, and the goal is to reconstruct it. In practical sequencing settings, the accessible measurements are limited to contiguous segments of the underlying strand; one cannot, for instance, simultaneously probe all odd positions at unit cost.
Likewise, in more general settings one might be restricted to a fixed family of allowed subsets $\mathcal{Q} \subseteq 2^{[n]}$, and the structure of $\mathcal{Q}$ can significantly affect the power of the algorithm.

In this work, we study the quantum query complexities of learning an unknown string when the allowed query subsets are drawn from an arbitrary fixed collection $\mathcal{Q}$. This unifies previously studied models, including the standard query model $\mathcal{Q} = \cbra{\cbra{i} : i \in [n]}$, and the search-with-wildcards model $\mathcal{Q} = 2^{[n]}$ under a single framework.

\subsection{Our Results}
We refer the reader to Section~\ref{sec:prelims} for formal definitions. We summarize here the notions that we need to state our main results. For an arbitrary collection of sets $\cQ \subseteq 2^{[n]}$ and a function $F : \pmone^n \to E$,
let $\Q^{\cQ}(F)$ denote the quantum query query complexity of computing $F$, where an algorithm can make unit-cost queries of the form $\bI[x_S = b_S]$ for an arbitrary $S \in \cQ$ and $b \in \pmone^S$ of its choice. We call $\Q^\cQ(F)$ the \emph{quantum $\cQ$-substring query complexity of $F$}. Let $\oplus(\cdot)$ denote the parity function, i.e., $\oplus(x) = \prod_{i = 1}^n x_i$ for $x \in \pmone^n$. A function $f : \pmone^n \to \bR$ is said to be \emph{odd} if $f(x) = -f(-x)$ for all $x \in \pmone^n$. We write $\Var(f) := \Var_{\mathbf{x} \sim \pmone^n}[f(\mathbf{x})]$ for the variance of the random variable $f(\mathbf{x})$, where $\mathbf{x}$ is sampled uniformly at random from $\pmone^n$.

Our main technical contribution is to show that the quantum $\cQ$-substring query complexity of parity is characterized by the optimal value of an abstract analytic optimization program. Below, $f_{S|b}$ refers to the function $f$ restricted to the bits in $\bar{S}$ set to $b$.
\begin{theorem}\label{thm:main}
    Let $n$ be a positive integer and $\cQ \subseteq 2^{[n]}$. 
    Then $\Q^\cQ(\oplus) = \Theta(\val_\cQ)$, where
    \begin{equation}
        \label{eq:dfn-valq}
        \val_\cQ := \left(\max_{\substack{f : \pmone^n \to \bR\\ f \textnormal{ is an odd function}}} \frac{\|f\|_\infty}{\max\limits_{\substack{S \in \cQ,\\ b \in \pmone^{\bar{S}}}}\sqrt{\Var(f_{S|b})}}\right).
    \end{equation}
\end{theorem}

In other words, this characterization can concisely be stated as:
\begin{tcolorbox}
    $\Q^\cQ(\oplus)$ is asymptotically equal to the maximum, over all odd functions $f : \pmone^n \to \bR$, of the ratio of $\|f\|_\infty$ to the largest standard deviation of $f$ in a subcube whose free variables are a valid query set.
\end{tcolorbox}

In order to show this, we use the characterization of quantum query complexity by the negative-weight adversary bound~\cite{HLS07, Rei11}. We show that for all $\cQ$, the corresponding adversary bound \emph{equals} $\val_\cQ$ as defined above.
We next analyze $\val_\cQ$, as defined in \Cref{eq:dfn-valq} for various interesting and natural settings of $\cQ$ such as bounded-size sets, contiguous blocks, prefixes, and only the full set.

\begin{theorem}\label{thm:valvalues}
    Let $n$ be a positive integer. Then, the following bounds hold for $\val_\cQ$ for various classes of $\cQ \subseteq 2^{[n]}$.
    \[
\begin{array}{|c|c|c|c|}
\hline
\text{Query Set} & \cQ & \val_{\cQ} & \text{Reference}\\
\hline
\hline
\text{Bounded-size sets} & \{ S \subseteq [n] : |S| \le k \} & \Theta\!\left(\dfrac{n} {\sqrt{k}}\right) & \text{Theorem~\ref{thm:valbounded}}\\
\hline
\text{Contiguous blocks} & \{ S = [i,j] : i \le j \in [n] \} & \vphantom{\overline{\widetilde{\Theta}}(n)} \widetilde{\Theta}(n) & \text{Theorem~\ref{thm:valcontig}} \\
\hline
\text{Prefixes} & \{ [1,i] : i \in [n] \} & \Theta(n) & \text{Theorem~\ref{thm:valprefix}}\\
\hline
\text{Only full set} & \{ [n] \} & 2^{(n-1)/2} & \text{Theorem~\ref{thm:valfull}}\\
\hline
\end{array}
\]
\end{theorem}

These results have immediate implications for the query complexity of the string-recovery problem in these query models. Let $\rec : \pmone^n \to \pmone^n$ be the (string recovery/learning) function defined by $\rec(x) = x$.
Using Theorem~\ref{thm:main}, Theorem~\ref{thm:valvalues}, and the Bernstein-Vazirani algorithm~\cite{BV97} to recover a bitstring in one query given query access to arbitrary parities of it, we obtain the following corollary.\footnote{The upper bounds for contiguous and prefix queries follow from the naive upper bound of querying one bit at a time (or an increasingly long prefix in the prefix query model) to learn an input string in $O(n)$ queries. The upper bound for full set queries follows from \cite{Gro96}.}

\begin{corollary}\label{cor:maininstantiations}
    Let $n$ be a positive integer. Then, the following bounds hold on the quantum $\cQ$-substring query complexity of learning an input string.
    \[
\begin{array}{|c|c|c|}
\hline
\text{Query Set} & \cQ & \Q^{\cQ}(\mathrm{rec}) \\
\hline
\hline
\text{Bounded-size sets} & \{ S \subseteq [n] : |S| \le k \} & \Theta\!\left(\dfrac{n}{\sqrt{k}}\right) \\
\hline
\text{Contiguous blocks} & \{ S = [i,j] : i \le j \in [n] \} & \vphantom{\overline{\widetilde{\Theta}}(n)}\widetilde{\Theta}(n) \\
\hline
\text{Prefixes} & \{ [1,i] : i \in [n] \} & \Theta(n)\\
\hline
\text{Only full set} & \{ [n] \} & \Theta(2^{n/2})\\
\hline
\end{array}
\]

\end{corollary}
In particular, setting $k = n$ in the first bound above tightens the $O(\sqrt{n} \log n)$ algorithm of Ambainis and Montanaro~\cite{AM14} for quantum search with wildcards, and improves it to the optimal $O(\sqrt{n})$. We note that the optimal $O(\sqrt{n})$ bound in this case has already been shown by Belovs~\cite{Bel15} using completely different techniques from ours. 
Additionally, our results in the prefix-query model prove that recovering any bitstring cannot be done with a sublinear number of prefix queries, which resolves an open question raised in~\cite{xu2023quantum}. The result in the only full set case recovers the well-known tight $\Omega(\sqrt{N})$ lower bound for searching for a marked element in a database with $N$ items~\cite{BBHT98} and the matching upper bound is given by \cite{Gro96}. 

A lower bound of $\Omega(n/\log(n))$, for learning the entire string, was shown by \cite{CILNTTY12} and \cite{BV23} when the query set is restricted to contiguous blocks. We recover this lower bound using our framework.

\subsection{Our Techniques}

We now discuss the techniques that we use to show Theorems~\ref{thm:main} and \ref{thm:valvalues}. 
We then make a comparison with related work.

\subsubsection{Sketch of Proof of Theorem~\ref{thm:main}}
Our analysis begins from the characterization of quantum query complexity by the negative-weight adversary bound~\cite{Rei11}. This characterization considers the quantum query complexity of a function~$F : D \to E$, with respect to a query set $\cR$, where every query $q \in \cR$ is a function mapping $D$ to a set of query outcomes. Reichardt~\cite{Rei11} expresses the query complexity of $F$ w.r.t.\ $\cR$, up to a constant factor, in terms of the optimum value of a certain semidefinite optimization program known as the general adversary bound, which we refer to as the \textit{primal adversary bound}, defined as
\begin{equation}
\label{eq:dfn-adv}
\ADV^{\cR,\pm}(F) := \max_{\Gamma \neq 0}\frac{\|\Gamma\|}{\max_{q \in \cR} \|\Gamma_q\|}.
\end{equation}
Here the maximization is over all matrices $\Gamma \in \bR^{D \times D}$ 
with $F(x) = F(y) \implies \Gamma[x,y] = 0$, and the maximization in the denominator is over all available queries $\cR$. For a query $q \in \cR$, the matrix $\Gamma_q$ is defined to be equal to $\Gamma$, with all those $[x,y]$ entries zeroed out where the outcome of the query $q$ is the same on inputs $x$ and $y$. This quantity was defined in~\cite{HLS07}, where it was shown to be a lower bound on quantum query complexity. Reichardt~\cite{Rei11} subsequently showed that the query complexity of $F$ with respect to query set $\cR$ is equal to $\ADV^{\cR,\pm}(F)$ up to a constant factor.

It is well-known (see e.g.,~\cite[Theorem~1.2]{Rei09}) that the negative-weight adversary bound can be written as a semidefinite program (SDP). We show in \Cref{lem:equivalence_of_formulations} that the feasible solutions of the maximization version of this SDP are directly related to the choices of $\Gamma$ in~\Cref{eq:dfn-adv}. It then follows that symmetry reductions on the SDP-level can be translated into additional symmetry constraints on the matrix $\Gamma$ in~\Cref{eq:dfn-adv}, without changing the optimal objective value. Note that attempting to naively apply such a symmetrization to the formulation in \Cref{eq:dfn-adv} does not immediately work since, for example, the numerator could potentially reduce after symmetrizing. This is why all our symmetrization is done at the SDP level.

For the specific case where $F = \oplus$ and $\cR = \{(S,b) : S \in \cQ, b \in \pmone^S\}$, with $\cQ \subseteq 2^{[n]}$, we subsequently characterize in \Cref{lem:bit_flip_group_symmetry,lem:xor_structure_of_feasible_solution} the symmetries in the primal adversary bound SDP, from which we deduce in \Cref{cor:adv_xor_matrix} that a maximizing $\Gamma$ can be assumed to be the communication matrix of an $\XOR$ function. That is, we can assume that $\Gamma = M_{f \circ \XOR}$ for some $f : \pmone^n \to \bR$, where $M_{f \circ \XOR}(x, y) := f(x \oplus y)$, i.e., the $(x,y)$'th entry only depends on the bitwise XOR of $x$ and $y$. It is well known that such matrices are simultaneously diagonalizable by the Hadamard matrix, and hence diagonal in the Fourier basis, which motivates analyzing the numerator and denominators of \Cref{eq:dfn-adv} using Fourier analysis.

It is known that the numerator of \Cref{eq:dfn-adv} in this case, which is the spectral norm of $M_{f \circ \XOR}$, equals the (appropriately scaled) maximum Fourier coefficient of $f$~\cite{BC02}. For the denominator, we analyze the sparsity structure of the matrix $\Gamma_q$, for any query $(S,b) = q \in \cR$. Through a series of calculations involving Fourier analysis of restricted Boolean functions, we show that the matrix $\Gamma_q$'s spectral norm equals the (appropriately scaled) maximum, over all fixings of bits in $\bar{S}$, standard deviation of the restriction of $f$ to the subcube that fixes these bits in $\bar{S}$. Putting everything together, we obtain $\val_\cQ = \ADV^{\cR,\pm}(\oplus)$, which together with Reichardt's characterization~\cite{Rei11} implies Theorem~\ref{thm:main}. 

We provide a schematic illustration of the proof overview for \Cref{thm:main} in \Cref{fig:outline_framework_buildup}.

\begin{figure}[htbp!]
\centering

\begin{tikzpicture}[
  >={Latex[length=2mm]},
  edge/.style={->, thick},
  box/.style={rounded corners=2pt, draw, thick, align=center, fill=blue!5, inner sep=3pt, minimum width=2cm, minimum height=8mm}
]

\node[box] (start) at (-6,0) {$\Q^{\cQ}(\oplus)$};
\node[box] (adv) at (0,0) {Optimal\ sol.\ to\\ $\ADV^{\cQ,\pm}(\oplus)$};
\node[box] (baradv) at (6,0) {Optimal\ sol.\ to\\ $\overline{\ADV}^{\cQ,\pm}(\oplus)$};
\node[box] (symmalt) at (6,-3) {Symmetrized sol.\ to\\ $\overline{\ADV}^{\cQ,\pm}(\oplus)$};
\node[box] (symm) at (0,-3) {Symmetrized sol.\ to \\ $\ADV^{\cQ,\pm}(\oplus)$};
\node[box] (valQ) at (-6,-3) {$\val_\cQ$};

\draw[dashed,<->] (start) to node[above] {\cite[Theorem~1.4]{Rei11}} node[below] {(see \Cref{thm:adv_equals_quantum_queries})} (adv);
\draw[<->] (adv) to node[above] {\Cref{lem:equivalence_of_formulations}} (baradv);
\draw[->] (baradv) to node[left] {\Cref{lem:bit_flip_group_symmetry,lem:xor_structure_of_feasible_solution}} (symmalt);
\draw[->] (symmalt) to node[above] {\Cref{lem:xor_structure_of_feasible_solution}} (symm);
\draw[->] (symm) to[bend right=20] node[above] {Numerator:~\Cref{lem:xornorm}} (valQ);
\draw[->] (symm) to[bend left=20] node[below] {Denominator:~\Cref{lem:blockmatrix,lem:diagonal-blocks,lem:diagonal-blpcks} and \Cref{thm:denominatornorm}} (valQ);

\draw[<->, dashed] (valQ) to node[left] {\Cref{thm:main}} (start);
\draw[<->] plot[smooth] coordinates {(adv.south) (-1,-1) (-5.25,-1.5) (valQ.north)};
\node[above] at (-3,-1.25) {$=$};
\end{tikzpicture}
\caption{A schematic illustration to obtain our framework, which enables us to transform the problem of computing $\Q^{\cQ}(\oplus)$, for an arbitrary collection of sets $\cQ \subseteq 2^{[n]}$, into the task of analyzing an abstract analytic optimization program $\val_{\cQ}$. The solid lines represent exact equalities, whereas the dashed lines represents characterizations up to constants.}
\label{fig:outline_framework_buildup}
\end{figure}
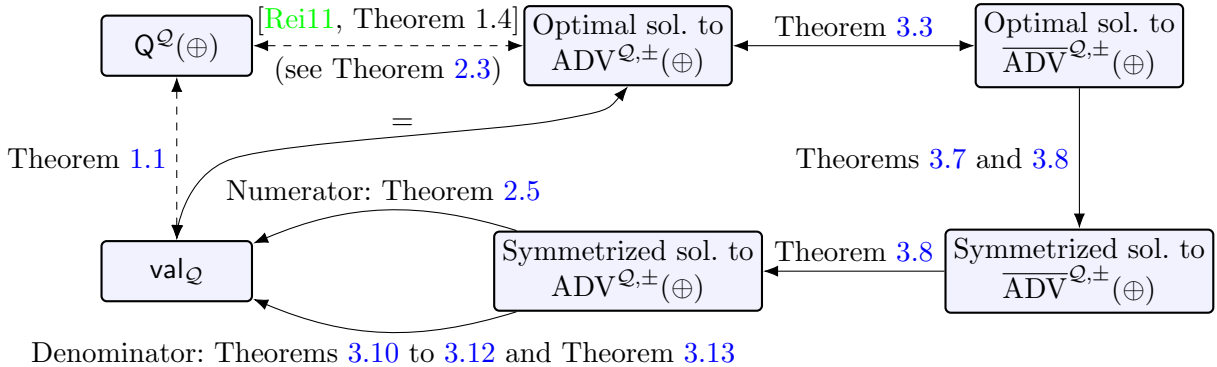

\subsubsection{Sketch of Proof of Theorem~\ref{thm:valvalues}}
Theorem~\ref{thm:valvalues} proves matching lower and upper  bounds on $\val_\cQ$ for four different classes of sets $\cQ$. Note that $\val_\cQ$ involves a maximization over all odd real-valued functions over the Boolean hypercube. The lower bounds are proven by exhibiting explicit functions for which every valid (for the class of $\cQ$ considered) subcube restriction has low standard deviation with respect to the maximum absolute value of the function. The upper bounds are proven by showing that every odd real-valued function over the Boolean hypercube has a high standard deviation (with respect to its maximum absolute value) when restricted to some subcube that is valid for the class of $\cQ$ considered. We now briefly discuss the proofs we give for various classes of $\cQ$.

\paragraph{Bounded size} Theorem~\ref{thm:valbounded} deals with query sets of bounded size. The valid subcubes are the ones that fix at least $n-k$ variables. The lower bound is witnessed by the function that maps each bit-string to its Hamming weight minus $n/2$ (note that this is an odd function). For the upper bound, we first prove by an induction on the number of variables, that for every odd function there is some subcube (not necessarily valid) which simultaneously has a high variance and the absolute values of all its degree-1 Fourier coefficients are large. If this subcube is not valid (i.e., if it fixes less than $n-k$ variables), we argue via the probabilistic method and Fourier analysis that the function can be further restricted by appropriately setting additional variables to a valid subcube, while simultaneously ensuring that the variance does not drop too much.

\paragraph{Contiguous blocks} Theorem~\ref{thm:valcontig} deals with query sets comprising consecutive variables (the same theorem statement holds true even if one is allowed wraparound). The valid subcubes are the ones that restrict consecutive variables. The upper bound of $n$ follows from the observation that the query sets include all singletons, and Proposition~\ref{prop:singletons} which verifies via the framework of Theorem~\ref{thm:main} the obvious fact that if all the single bit queries are allowed, then the quantum query complexity is $O(n)$.

The lower bound is witnessed by a function $f$ defined as follows. Given an input bit string of size $n$, view it as consisting of $n/k$ disjoint blocks of size $k$, where $k=\log n - \log \log n +O(1)$. Let the blocks be denoted by $z^{(1)},\ldots, z^{(n/k)}$. First, each $z^{(i)}$ is subjected to a function $g$ that maps the string $1^k$ to $1$, $(-1)^k$ to $-1$ and maps all other strings to $0$. Finally $f$ is defined to be $\sum_{i=1}^{n/k} g(z^{(i)})$. Note that each subcube setting contiguous variables (even if wraparound is allowed) touches, i.e., the endpoints of the block of queried indices are contained in, at most two of the blocks, so the restricted function is essentially a reduced version of $f$ on a smaller number of blocks plus a bounded function on at most $2k$ bits. Thus, up to an additive shift by a bounded function, the (random variable corresponding to the) restriction of $f$ is essentially the sum of independent random variables taking values in $\{0, 1, -1\}$ (the independence is due to the disjointness of the blocks), each of which is 0 except with probability $1/2^{k-1}$. Furthermore, the additional bounded summand is also independent of this sum and has $O(1)$ variance. We can thus apply linearity of variance for independent random variables to finish our argument.

\paragraph{Prefixes} Theorem~\ref{thm:valprefix} deals with query sets which are prefixes, i.e., contain the first $i$ variables for some positive integer $i$. The theorem shows that for this class $\cQ$ of queries $\val_\cQ$ is $\Theta(n)$. Valid subcubes in this setting fix a set of variables that is a suffix, i.e., contains the last $i$ variables for some non-negative integer $i$.

The lower bound is witnessed by a function that is defined first by constructing a \emph{decision list}: if $x_{n-1} = -1$ then output 1, else if $x_{n-2}=-1$ then output 2, \dots, else if $x_1 = 1$ then output $n-1$, else output $n$.
We then turn this decision list into an odd function by reflecting it about an additional appended variable and inverting the reflection (both the input variables and the function output). The maximum value of this function is $n$. At a high level, inside subcubes that only fix suffixes, the values that the function takes `grow linearly' but their probabilities `drop exponentially'. This allows us to show that every valid subcube has $O(1)$ variance.

For the upper bound we consider an arbitrary odd real function $f$, and consider the sequence of subcubes that are obtained by fixing the last $i$ variables consistent with an input that maximizes $f$, for $i=0, 1, \ldots, n$. Note that the average of $f$ over the first subcube (which is $\pmone$) is $0$ (as the function is odd and hence balanced), and over the last subcube is $\|f\|_\infty$. Thus, we have that there is a point where the function average jumps by at least $\|f\|_\infty/n$. This implies that switching the value of the variable that is being restricted at that step leads to a large change in the function average. We could quantify this phenomenon in the form of a lower bound of $\|f\|_\infty/n$ on the absolute value of the degree-1 Fourier coefficient of the restricted function (note that potentially multiple variables have been fixed before we reach this point) corresponding to this variable. A lower bound on the variance, and hence an upper bound of $\val_\cQ$, then follows easily by using the Fourier formula for variance.
\begin{remark}
In \Cref{sec:prefix_dictator} we provide a proof directly using the negative-weight adversary bound without resorting to our framework, of the stronger result that computing the last bit of an input string (which can only be an easier task than learning the whole string) requires $\Omega(n)$ prefix queries. We choose to still retain \Cref{thm:valprefix} and its proof to demonstrate the versatility of our framework.
\end{remark}

\paragraph{Only full set}
\Cref{thm:valfull} deals with the case when $\cQ$ consists of the full set of variables. In this case the denominator in the definition of $\val_\cQ$ is just the standard deviation of the function over all inputs.
For the lower bound, we construct a function $f$ that takes only two non-zero values: value $+1$ on the all-$1$ input, value $-1$ on the all-$-1$ input, and $0$ everywhere else. From this, its standard deviation turns out to be extremely small, exactly $2^{-(n-1)/2}$, giving the lower bound.

For the upper bound, consider any odd function $f$. Because $f$ is odd, it attains values $\|f\|_\infty$ and $-\|f\|_\infty$ on some pair of opposite inputs. These two points alone already contribute enough to the second moment to guarantee that the standard deviation in the whole cube is at least $\|f\|_\infty \cdot 2^{-(n-1)/2}$. Dividing $\|f\|_\infty$ by this quantity, as required in the definition of $\val_\cQ$, shows that no odd function can achieve a value above about $2^{(n-1)/2}$. This gives the upper bound.

Together, the lower and upper bounds match, so the value for this setting of $\cQ$ is exactly $2^{(n-1)/2}$.

\subsubsection{Comparison With Related Work}
The quantum $\cQ$-substring query complexity of $\rec$ with $\cQ = 2^{[n]}$ was introduced and studied by Ambainis and Montanaro~\cite{AM14}, who dubbed this problem the ``search with wildcards'' problem and gave an $O(\sqrt{n} \log n)$ upper bound and a near matching $\Omega(\sqrt{n})$ lower bound. The algorithm by Ambainis and Montanaro uses the Pretty Good Measurement to iteratively learn $O(\sqrt{n})$ bits and fixes any errors made via binary search. In the search with wildcards setting, binary search can be performed with just $O(\log n)$ queries. Therefore, the algorithm learns the entire string in $O(\sqrt{n}\log n)$ queries. Setting $k = n$ in our Corollary~\ref{cor:maininstantiations} gives a tight $O(\sqrt{n})$ bound for this task.
Ambainis and Montanaro also gave quantum query algorithms for the combinatorial group testing problem, which they observed to be a generalization of the $\cQ$-substring query complexity of $\rec$ with $\cQ = 2^{[n]}$. Subsequently, Belovs~\cite{Bel15} gave, among other results, tight bounds on the quantum query complexity of combinatorial group testing, which implies a tight $O(\sqrt{n})$ upper bound on the quantum $2^{[n]}$-substring query complexity of $\rec$. In contrast with the algorithms of Ambainis and Montanaro, Belovs showed a $O(\sqrt{n})$ upper bound on the dual of the negative-weight adversary bound. Since the negative-weight adversary bound is known to be asymptotically equal to quantum query complexity~\cite{Rei11}, this immediately gives an upper bound on $\Q^{\cQ}(\rec)$.

In contrast to the work of Ambainis and Montanaro~\cite{AM14} who construct an explicit algorithm, and that of Belovs~\cite{Bel15} who constructs an explicit solution to the dual adversary bound, we deal, perhaps counter-intuitively, completely with the \emph{primal} version of the negative-weight adversary bound~\cite{HLS07}. Our work does have the drawback that it does not give explicit algorithms for any of our upper bounds. However, as far as we are aware, ours is the first work to prove novel quantum query \emph{upper} bounds using the \emph{primal} version of the negative-weight adversary bound (which is primarily used as a lower bound technique) without explicitly resorting to SDP duality.

When $\cQ$ is the set of all contiguous blocks, a lower bound of $\Omega(n/\log(n))$ was previously shown by \cite[Lemma 7]{CILNTTY12} via the positive-weight adversary method. The same lower bound, even for computing Parity, was also shown by \cite[Theorem~19]{BV23}, where they lower bound the approximate degree and obtain the lower bound via the polynomial method \cite{BBCMdW01}. Furthermore, their lower bound holds even in the setting where the success probability of the algorithm is inverse polynomially close to 1/2. In this work, we obtain the same $\Omega(n/\log(n))$ query lower bound (in the bounded-error setting) using our framework, using techniques different from the previous works.

Another work of relevance to ours is~\cite{BBGK18}, who showed classical lower bounds for composed functions using quantum upper bounds for combinatorial group testing~\cite{Bel15}. It is not immediately clear to us if our results have any new consequences for classical lower bounds using their framework, but our work can be viewed as partial progress towards answering their question~\cite[Section~V]{BBGK18} ``Are there nontrivial quantum algorithms that use $g$-queries to learn $x$ for some function $g \notin \cbra{\OR, \XOR}$? Are there nontrivial quantum algorithms
that use $g$-queries to compute some other function $h(x)$ of the input?'' Theorem~\ref{thm:main} gives a complete framework that allows one to easily analyze the complexity of computing Parity (which is often equivalent to the complexity of learning the full input string) when $g$ is restricted to subcube queries. Our framework gives provably tight bounds in this setting for all classes of allowed subcubes to query (with the constraint that if the algorithm is allowed to query if $x$ is in a subcube whose fixed variables are $S$, all fixing of variables in $S$ are allowed queries). The key advantage of our framework is that it replaces quantum and SDP-based reasoning with a simple, combinatorial analysis of real-valued functions on the Boolean hypercube. We justify this simpler viewpoint via Theorem~\ref{thm:valvalues}, which establishes tight bounds within this framework for several natural classes of allowed subcube query sets.

\subsubsection{Outlook}

While our results do not apply in the combinatorial group testing setting, where only OR queries are allowed and the input string is promised to have bounded Hamming weight, we consider a generalization of the search with wildcards problem in a different direction, namely when a query algorithm is restricted on the support of the substring queries it can make.
Our results are specifically of interest in the context of biological applications such as genome reconstruction, DNA sequencing, and RNA sequencing: In these models, a unit-cost task is often assumed to be observing a contiguous substring of the input string~\cite{SS95, robertson2010novo, CWE15, LD24}. Our second bound in Corollary~\ref{cor:maininstantiations} shows that quantum algorithms cannot provide an advantage over classical ones in this setting.

We view our work as an invitation to explore upper bounds through the \emph{primal} adversary formulation itself, by studying the structural properties that optimal adversary matrices must satisfy. While this formulation can tend to be regarded as less convenient than its dual, our results illustrate that, with sufficient symmetrization and structural insight, the primal characterization can in fact yield tight bounds and reveal new patterns in how optimal $\Gamma$ matrices arise. This suggests that it may be a more practical tool for upper bounds than previously thought.

\section{Preliminaries}\label{sec:prelims}

All logarithms in this paper are taken base 2. For a positive integer $n$ and integers $i \leq j \in [n]$, we use the notation $[i, j] := \cbra{k \in [n] : i \leq k \leq j}$. For strings, $x, y \in \pmone^n$, we use the notation $x \oplus y$ to denote the bitwise XOR of $x$ and $y$, or equivalently, the string obtained by taking component-wise multiplication of $x$ and $y$. For a string $x \in \pmone^n$, we use the notation $|x|$ to denote the Hamming weight of $x$, i.e., the number of $-1$'s in $x$.
For a string $x \in \pmone^n$ and $S \subseteq [n]$, we use $x_S$ to denote the string $x$ restricted to the bits in $S$. For an expression $X$, we define $\bI[X]$ to equal $1$ if $X$ is true, and $0$ if $X$ is false. For a string $x = (x_1, x_2, \dots, x_n) \in \pmone^n$, we use the notation $-x$ to denote the string $(- x_1, - x_2, \dots, - x_n)$. We denote vectors with the ket-notation, $\ket{v}$, and refer to their conjugate transpose with the bra-notation, i.e., $\bra{v}$.

A function $f : \pmone^n \to \bR$ is said to be an \emph{odd function} if $f(x) = -f(-x)$ for all $x \in \pmone^n$. We abuse notation and use $f$ to interchangeably mean a function as well as a random variable corresponding to picking a value of $f$ uniformly at random. That is, for $f : \pmone^n \to \bR$, we use $\E[f] := \E_{x \in \pmone^n}[f(x)]$, $\E[f^2] := \E_{x \in \pmone^n}[f(x)^2]$, and $\Var(f) := \E[f^2] - \E[f]^2$. We use $\|f\|_\infty$ to denote $\max_{x \in \pmone^n}|f(x)|$. For a positive integer $n$ and a $2^n$-dimensional real matrix $M$ whose rows and columns are indexed by $n$-bit strings, we say that a $M$ is the \emph{communication matrix of an $\XOR$ function} if there exists a function $f : \pmone^n \to \bR$ such that $M(x, y) = f(x \oplus y)$. We use the notation $M_{f \circ \XOR}$ to denote the communication matrix of $f \circ \XOR$.

A collection of sets $\cQ \subseteq 2^{[n]}$ is said to be \emph{downward closed} if for all $Q \in \cQ$, all subsets of $Q$ are also in $\cQ$.

For an $n$-dimensional real square matrix $M$, its spectral norm, denoted $\|M\|$, is defined by
\[
\|M\| := \sup_{x \in \bR^n : \norm{x}_2 = 1}\norm{Mx}_2,
\]
where $\norm{\cdot}_2$ denotes the usual Euclidean 2-norm. Equivalently it is well known that the following characterization also holds true.
\begin{fact}\label{fact:spectral}
    For an $n$-dimensional real square matrix $M$,
    \[
    \|M\| = \sqrt{\lambda_{\max}(M^T M)},
    \]
    where $\lambda_{\max}(\cdot)$ equals the maximum eigenvalue.
\end{fact}

\subsection{Quantum Query Complexity}

Let $\cal{H}$ be a finite-dimensional Hilbert space. Any quantum algorithm $\cal{A}$ can be represented as a sequence of unitary operators $\{U_i\}_{i=0}^{T} \in \cal{L}(\cal{H})$ on an initial state $\ket{\psi_0} \in \cal{H}$. The algorithm is terminated by a measurement on $\ket{\phi} = U_TU_{T-1} \cdots U_1U_0\ket{\psi_0}$ to observe an outcome.

Let $D$ and $E$ be finite sets and $F:D\rightarrow E$ be a function. For every $x \in D$, let $O_x$ be a unitary operator that depends on $x$, referred to as the \textit{oracle}. A quantum query algorithm $\cal{A}$ alternates input-independent unitary operators with calls to the oracle, i.e., it applies the operations $U_0$, $O_x$, $U_1$, $O_x$, etc. The oracle calls are referred to as \textit{queries} to the input $x$. The algorithm $\cal{A}$ is a bounded-error, quantum query algorithm for computing $F$ if $\forall x \in D$ the output of $\cal{A}$ is $F(x)$ with probability at least $2/3$. The quantum query complexity of computing $F$ is the minimum number of queries required to compute $F$ with probability $2/3$. This is denoted by $\Q(F)$.

In this work, we consider algorithms with access to different \emph{query sets} and therefore use $\Q^{\cal{R}}(F)$ to refer to the quantum query complexity of computing $F$ using query set $\cal{R}$. Formally a query set $\cal{R}$ is a set consisting of functions $r : \pmone^n \to \cal{O}$ for some discrete range $\cal{O}$. The usual query model is captured by the query set $\cal{R} = \cbra{r_i : \pmone^n \to \pmone \mid i \in [n]}$ where $r_i(x) := x_i$. The final state of a quantum query algorithm with query set $\cR$ looks like 
\begin{equation*}
    \ket{\phi} = U_T O_x \ldots U_2 O_x U_1 O_x U_0\ket{\psi_0}
\end{equation*}
where the action of $O_x$ on each basis state is
\begin{equation*}
    O_x : \ket{r}\ket{b}\rightarrow \ket{r}\ket{b + r(x) \mod |\cal{O}|} \qquad \forall r\in \cR, b \in \cbra{0, 1, \dots, |\cal{O}|}.
\end{equation*}

Let $x\in\{-1,1\}^n$ be an input string and $\cal{Q}\subseteq 2^{[n]}$. We abuse notation and use $\cal{Q}$ to be the query set consisting of all queries $q = (S,b_S)$ of the form $q(x) = \mathbb{I}[x_S = b_S]$ where $S\in\cal{Q}$, $b_S\in\{-1,1\}^S$ and $x_S$ is the substring of $x$ restricted to indices in $S$. Throughout this work, we only work with queries of this form. 
Thus, for $\cQ \subseteq 2^{[n]}$, we abuse notation and use $\cQ$ to also denote the query set $\cbra{q = (S, b) : S \in \cQ, b \in \pmone^S}$.

\begin{defi}[Positive semi-definiteness]
    A matrix $\Gamma \in \bR^{N \times N}$ is said to be \emph{positive semi-definite}, denoted by $\Gamma \succeq 0$, if 
    \begin{equation*}
        \bra{v}\Gamma\ket{v} \ge 0 \qquad \forall \ket{v}\in\mathbb{R}^N.
    \end{equation*}
\end{defi}

\paragraph*{Standard formulation of primal adversary bound}

The primal adversary bound is a semi-definite program that characterizes the quantum query complexity of computing function $F:D\rightarrow E$, relative to a query set $\cal{R}$. To that end, $\forall q \in \cal{R}, \Delta_q \in \{0,1\}^{D\times D}$ is the matrix defined as:
\[
\Delta_q[x, y] = \begin{cases}
    1, & q(x) \neq q(y), \\
    0, & \textnormal{otherwise}.
\end{cases}
\]
That is, $\Delta_q[x,y] = 1$ if and only if the query $q$ can distinguish between inputs $x$ and $y$. Throughout this work, we use the notation $\ADV^{\cal{R},\pm}$ to denote the adversary bound relative to query set $\cR$. Now we state the primal version of the general adversary bound, which we will refer to as the \textit{primal adversary bound}.

\begin{tcolorbox}[title=Primal adversary bound]
    \begin{align}\label{eqn:primal_adversary_bound_standard}
    \ADV^{\cal{R}, \pm}(F) = & \max_{\Gamma}\quad \frac{\|\Gamma\|}{\max\limits_{q\in{\cal{R}}}\|\Gamma\circ\Delta_q\|} \\
    \nonumber\text{subject to}\quad & \Gamma \in \mathbb{R}^{D\times D},\\
    \nonumber& \Gamma \text{ is symmetric},\\
    \nonumber& \Gamma[x,y] = 0.\tag{$\forall \; x,y\in D \text{ with } F(x) = F(y)$}
\end{align}
\end{tcolorbox}

Since the objective function is invariant under multiplying $\Gamma$ with a constant factor, we can without loss of generality assume that the denominator is equal to $1$. Thus, we can equivalently write the SDP computing $\ADV^{\cal{R}, \pm}$ with objective function $\max_\Gamma\|\Gamma\|$, and the additional constraint that $\forall q\in\cal{R},\|\Gamma\circ\Delta_q\|\le 1$.

Reichardt~\cite{Rei11} showed that the primal adversary bound is a tight characterization of quantum query complexity.

\begin{theorem}[{\cite[Theorem~1.4]{Rei11}}]\label{thm:adv_equals_quantum_queries}
    Let $D$ and $E$ be finite sets and $F:D\rightarrow E$ be a function. The bounded-error, quantum query complexity of computing $F$ using query set $\cal{R}$ is characterized by the adversary bound.
    \begin{equation*}
        \Q^{\cal{R}}(F) = \Theta(\ADV^{\cal{R}, \pm}(F)).
    \end{equation*}
\end{theorem}

Although \cite[Theorem~1.4]{Rei11} states the adversary bound for index queries, it can be observed that the adversary bound characterizes the quantum query complexity for any query set and therefore we present it as such. 
Let $\cR$ be a query set (assume for simplicity that all queries output values in $\pmone$. The argument below can easily be generalized if this is not the case). For every $x \in D$, define the string $\cR(x) = (r(x))_{r \in \cR}$. Define the partial function $G : \pmone^\cR \to E$ on the input domain $\cbra{z \in \pmone^\cR : z = \cR(x) \textnormal{ for some } x \in D}$ as $G(\cR(x)) := F(x)$. Observe that querying a input variable, say $R \in \cR$, to $G$ on input $\cR(x)$ is exactly the same as making the query $R$ on input $x$. Thus, $\Q^{\cR}(F) = \Q(G)$ and $\ADV^{\cR,\pm}(F) = \ADV^{\pm}(G)$, and therefore the result follows by applying \cite[Theorem~1.4]{Rei11} to $G$.

When $\cQ \subseteq 2^{[n]}$, and the query set consists of all queries of the form $\mathbb{I}[x_S = b]$ where $S \in \cQ$ and $b \in \pmone^S$ are arbitrary, recall that we abuse notation and use $\cQ$ to also denote the query set $\cbra{q = (S, b) : S \in \cQ, b \in \pmone^S}$. Thus $\ADV^{\cQ, \pm}(\cdot)$ and $\Q^{\cQ}(\cdot)$ denote the negative-weight adversary bound and the quantum query complexity, respectively, with respect to this query set.

\subsection{Fourier Analysis of Boolean Functions and Restrictions}

Below are some relevant formal definitions and Fourier-analytic properties of restricted functions, adapted from~\cite{o2014analysis}.

Consider the vector space of functions from $\pmone^n$ to $\bR$, equipped with the following inner product.
\[
\langle f, g \rangle := \E_{x \in \pmone^n}[f(x)g(x)] = \frac{1}{2^n}\sum_{x \in \pmone^n}f(x)g(x).
\]
For a set $S \subseteq [n]$, define the parity function $\chi_S : \pmone^n \to \bR$ by $\chi_S(x) = \prod_{i\in S} x_i$. It is not hard to show that the set of all parities $\cbra{\chi_S : S \subseteq [n]}$ forms an orthonormal basis of the space above. Thus, every function $f : \pmone^n \to \bR$ has a unique expansion with respect to this basis, $f = \sum_{S \subseteq [n]} \widehat{f}(S)\chi_S$. This expansion is called the Fourier expansion of $f$ and the coefficients $\cbra{\widehat{f}(S) : S \subseteq [n]}$ are called the Fourier coefficients of $f$.
\begin{fact}[Parseval's Theorem]
\label{fact:parseval}
Let $f:\pmone^n \to \bR$ be a function. Then
\[\E_{x \in \pmone^n}f(x)^2=\langle f,f\rangle=\sum_{S \subseteq [n]}\widehat{f}(S)^2.\]
\end{fact}

\begin{lemma}[{\cite{BC02}, also see~\cite[Lemma~2.2.3]{M18}}]\label{lem:xornorm}
    Let $f : \pmone^n \to \bR$ be a function and let $M$ be the communication matrix of $f \circ \XOR$. Then,
    \[
        \|M\| = 2^n \max_{S \subseteq [n]} |\widehat{f}(S)|.
    \]
\end{lemma}

Specifically, we require the following fact. Below, $H$ represents the $2^n$-dimensional Hadamard matrix defined by $H(S, T) = \frac{1}{2^{n/2}}(-1)^{|S \cap T|}$. 
\begin{fact}\label{fact:xorhadamard}
    Let $f : \pmone^n \to \bR$ be a function. Then,
    \[
        M_{f \circ \XOR} = H D_{2^n\widehat{f}}H^{-1},
    \]
    where $D_{2^n\widehat{f}}$ is a diagonal matrix with $D_{2^n\widehat{f}}(S) = 2^n\widehat{f}(S)$ for all $S \subseteq [n]$.
\end{fact}

\begin{defi}
    Let $f : \pmone^n \to \mathbb{R}$ be a function, and let $(J, \bar{J})$ be a partition of $[n]$. Let $z \in \pmone^{\bar{J}}$. Then we write $f_{J|z} : \pmone^J \to \mathbb{R}$ (pronounced ``the restriction of $f$ to $J$ using $z$'') for the subfunction of $f$ obtained by fixing the coordinates in $\bar{J}$ to the bit values $z$, i.e.,
    \begin{equation*}
        \forall x\in\pmone^J, f_{J|z}(x) := f(x_Jz_{\bar{J}})
    \end{equation*}
    where $x_Jz_{\bar{J}}$ represents the string that equals $x$ on $J$ and $z$ on $\bar{J}$.
\end{defi}

\begin{proposition}[{\cite[Proposition 3.21]{o2014analysis}}]\label{prop:restcoeffs}
    Let $f : \pmone^n \to \mathbb{R}$ and let $(J, \bar{J})$ be a partition of $[n]$. For $z \in \pmone^{\bar{J}}$ and $S \subseteq J$ we have
    \[
        \widehat{f_{J|z}}(S) = \sum_{T \subseteq \bar{J}}\wh{f}(S \cup T)\chi_T(z).
    \]
\end{proposition}

Setting $z = 1^{\bar{J}}$ in the proposition above, we obtain the following corollary.

\begin{corollary}\label{cor:0restcoeffs}
    Let $f : \pmone^n \to \mathbb{R}$ and let $(J, \bar{J})$ be a partition of $[n]$. For $S \subseteq J$ we have
    \[
        \wh{f_{J|1^{\bar{J}}}}(S) = \sum_{T \subseteq \bar{J}}\wh{f}(S \cup T).
    \]
\end{corollary}

We also obtain the following general corollary.
\begin{corollary}[{\cite[Corollary 3.22]{o2014analysis}}]\label{cor:parsevalrestcoeffs}
    Let $(J, \bar{J})$ be a partition of $[n]$ and fix $S \subseteq J$. Suppose $z \sim \pmone^{\bar{J}}$ is chosen uniformly at random. Then,
    \[
        \E_z[\wh{f_{J|z}}(S)^2] = \sum_{T \subseteq \bar{J}} \wh{f}(S \cup T)^2.
    \]
\end{corollary}

\begin{lemma}\label{lem:fourierequiv}
    For $f : \pmone^n \to \bR$, the following conditions are equivalent:
    \begin{itemize}
        \item For all $x$ with $|x|$ even, $f(x) = 0$.
        \item For all $S \subseteq [n]$, we have $\widehat{f}(S) = - \widehat{f}(\bar{S})$.
    \end{itemize}
\end{lemma}

\begin{proof}
    For all $x \in \pmone^n$,
    \begin{align*}
        f(x) = \sum_{S\subseteq[n]}\widehat{f}(S)\chi_S(x) = \frac{1}{2}\sum_{S\subseteq[n]}\left[ \widehat{f}(S)\chi_{S}(x) + \widehat{f}(\bar{S})\chi_{\bar{S}}(x) \right].
    \end{align*}

    Assume that $\forall S\subseteq[n]$, $\widehat{f}(S) = -\widehat{f}(\bar{S})$. Let $x\in\{-1,1\}^n$ such that $|x|$ is even. Now we note that if $|x_S|$ is even, then $|x_{\bar{S}}|$ is also even. Similarly, when $|x_S|$ is odd, $|x_{\bar{S}}|$ is also odd. This implies $\chi_S(x) = \chi_{\bar{S}}(x)$ for all $x$ with even Hamming weight. Thus, for such $x$ we have
    \begin{align*}
        f(x) = \frac{1}{2}\sum_{S\subseteq[n]}\left[ \widehat{f}(S) + \widehat{f}(\bar{S}) \right]\chi_S(x) = \frac{1}{2}\sum_{S\subseteq[n]}\left[ \widehat{f}(S) - \widehat{f}(S) \right]\chi_S(x) = 0.
    \end{align*}

    To prove the other direction, we begin with the assumption that $f(x) = 0$ for all $x\in\{-1,1\}^n$ such that $|x|$ is even. By orthogonality of the characters, we have
    \begin{align*}
        \widehat{f}(S) = \E[f(x)\chi_S(x)] = - \E[-f(x)\chi_S(x)] = - \E[f(x)\chi_S(x)\chi_{[n]}(x)] = -\E[f(x)\chi_{\bar{S}}(x)] = -\widehat{f}(\bar{S}),
    \end{align*}
where the third equality holds since for all $x$ in the support of the expectation, we have $|x|$ odd and therefore $\chi_{[n]}(x) = -1$, and the fourth equality uses $\chi_{A}(x) \cdot \chi_{B}(x) = \chi_{A \triangle B}(x)$ where $A \triangle B$ denotes the symmetric difference of $A$ and $B$.
\end{proof}

We require the following well-known expression for the variance of a function. See, for example,~\cite[Proposition 1.13]{o2014analysis} for a simple proof.
\begin{lemma}\label{lem:fouriervariance}
    Let $f : \pmone^n \to \bR$ be a function. Then,
    \[
    \Var(f) = \sum_{S \neq \emptyset}\widehat{f}(S)^2.
    \]
\end{lemma}

Finally, we recall this fact.
\begin{fact}
    \label{fact:off-diagonal-matrix-norm}
    Let $v \in \mathbb{R}^d$ be a vector, and
    \[A = \begin{bmatrix}
        0 & v^T \\
        v & 0
    \end{bmatrix} \in \mathbb{R}^{(d+1) \times (d+1)}.\]
    Then, $\norm{A} = \norm{v}_2$, the Euclidean norm of $v$. Here, the bottom-right 0 in $A$ is a $d \times d$ zero matrix.
\end{fact}

\section{The Framework}\label{sec:framework}

In this section we build our framework that allows us to analyze the $\cQ$-substring query complexity of learning input strings, for arbitrary $\cQ$.

\subsection{Learning vs.~Parity}
We first make the simple observation that for every collection of sets $\cQ \subseteq 2^{[n]}$, the quantum query complexity of learning $x$ is bounded from below by the quantum query complexity of computing the parity of all of the bits of $x$.

\begin{observation}\label{obs:paritylowerbound}
    Let $\cQ \subseteq 2^{[n]}$. Then $\Q^\cQ(\rec) \geq \Q^\cQ(\oplus)$.
\end{observation}
This immediately follows from the fact that once a string has been learned, the parity of its input bits can be computed with no extra queries. We remark that the same argument in fact yields $\Q^\cQ(\rec) \geq \max_{S \subseteq [n]}\Q^\cQ(\oplus_S)$, but we do not use this.

Next we observe that in the case when $\cQ$ is downward closed, these complexities are actually equal.
\begin{lemma}\label{lem:downwardclosed}
For all $\cQ \subseteq 2^{[n]}$ that is downward closed, we have $\Q^\cQ(\rec) = \Q^\cQ(\oplus)$.
\end{lemma}
\begin{proof}
    \Cref{obs:paritylowerbound} provides the lower bound on $\Q^{\cQ}(\rec)$, so it remains to prove the upper bound. Note that we can recover a bit string if we are able to compute the parity on each of the subsets in superposition, by the Bernstein-Vazirani algorithm. Thus, it remains to prove that $\Q^{\cQ}(\oplus_S) \leq \Q^{\cQ}(\oplus)$, for all $S \subseteq [n]$.

    To that end, let $S \subseteq [n]$. For any $x \in \pmone^n$, let $y(x) \in \pmone^n$ be the input that is equal to $x$ on $S$, and identically $1$ on $\bar{S}$. Then, $\oplus_S(x) = \oplus(y(x))$, so it remains to run the algorithm on $y(x)$. To that end, note that we can simulate a substring query $(T,b)$ to $y(x)$ as follows. First, we check if $b$ is identically $1$ on $T \setminus S$, else we simply output $1$ (False). Otherwise, we perform the query $(S \cap T, b_S)$ on $x$, and we observe that $y(x)_T = b$ if and only if $y(x)_{S \cap T} = b_S$. Since $\cQ$ is downward-closed and $S \cap T \subseteq T$, we indeed have access to this query, and so we can simulate the parity-algorithm on $y(x)$ without additional overhead.
\end{proof}

We conclude that in order to characterize the $\cQ$-complexity of bit string recovery, it suffices to characterize the $\cQ$-complexity of computing the parity function, as long as $\cQ$ is downward closed.

\subsection{Simplification by Symmetry Reduction}

In this subsection we perform a symmetry reduction on the primal adversary bound for the quantum $\cQ$-substring quantum query complexity of computing parity, for an arbitrary $\cQ \subseteq 2^{[n]}$. We show that the maximizing matrix $\Gamma$ can be assumed to be the communication matrix of a XOR function, i.e., we may assume that $\Gamma = M_{f \circ \XOR}$ for some $f : \pmone^n \to \bR$. Recall that $M_{f \circ \XOR}[x, y] = f(x \oplus y)$ for all $x,y\in\{-1,1\}^n$.

\paragraph{Alternate formulation of adversary bound for functions with Boolean output}

Here we use an alternative formulation of the adversary bound that was shown by Cornelissen~\cite[Theorem~6.2.4]{Cor23}. This formulation allows us to average over optimal solutions without any loss in the optimal value. An important detail to note is that this formulation holds only for functions with Boolean output.\footnote{A similar optimization program can be derived for functions with non-Boolean output. Since its presentation is slightly more cumbersome and we don't need it in this work, we merely present the Boolean version here.}

\begin{tcolorbox}[title=Alternate formulation of primal adversary bound for functions with Boolean output]
    \begin{align}\label{eqn:primal_adversary_bound_new}
    \overline{\ADV}^{\cal{R}, \pm}(F) = \quad\max_{\Gamma, \beta}\quad & \sum_{x,y\in D} \Gamma[x,y]\\
    \text{subject to}\quad & \Gamma \in \mathbb{R}^{D\times D}, \beta \in \mathbb{R}_{\geq0}^D,\nonumber\\
    \nonumber& \Gamma \text{ is symmetric},\nonumber\\
    \nonumber& \Gamma[x,y] = 0,\tag{$\forall \; x,y\in D \text{ with } F(x) = F(y)$}\\
    \nonumber& \mathsf{diag}(\beta) - \Gamma\circ\Delta_q \succeq 0,\tag{$\forall q\in \cal{R}$}\\
    \nonumber& \sum_{x\in f^{-1}(1)}\beta[x] = \sum_{x\in f^{-1}(-1)}\beta[x] = \frac{1}{2}.
\end{align}
\end{tcolorbox}

\begin{lemma}\label{lem:equivalence_of_formulations}
    Let $D\subseteq\{-1,1\}^n$ and $F:D\rightarrow\{-1,1\}$ be a Boolean function. Then $\ADV^{\cR,\pm}(F) = \overline{\ADV}^{\cR,\pm}(F)$. Moreover, if $(\beta,\Gamma)$ is an optimal solution for $\overline{\ADV}^{\cR,\pm}(F)$, then $\Gamma'$, defined as\footnote{Here, we use the convention that $0/0 = 0$.}
    \[\Gamma'[x,y] = \frac{\Gamma[x,y]}{\sqrt{\beta[x]\beta[y]}},\]
    is an optimal solution for $\ADV^{\cR,\pm}(F)$.
\end{lemma}

The proof of this lemma essentially appears in the proof of~\cite[Theorem~6.2.4]{Cor23}, but not with the exact formulation above. We provide a proof of \Cref{lem:equivalence_of_formulations} in \Cref{sec:equivalence_of_formulations_proof} for completeness.
The benefit of this alternate formulation of the primal adversary bound is that it is a convex optimization problem. Therefore, it is always possible to take a linear combination of two feasible solutions, say with objective values $c_1$ and $c_2$, and obtain a new feasible solution with an objective value that is at least as big as $\min\{c_1,c_2\}$. We will use this to perform a symmetry reduction on the level of the SDP. In particular, we show how symmetries in the function and in the query set translate to additional constraints that we can impose on the optimal solutions to the alternate formulation of the primal adversary bound.

Let $\sigma : D\rightarrow D$ be a permutation of the domain $D$. We abuse notation to use $\sigma$ to refer to the permutation function as well as the corresponding permutation matrix. It will be clear based on context whether $\sigma$ is treated as a function or a matrix. When treated as a matrix, right-multiplying by $\sigma^T$ corresponds to permuting columns according to the permutation $\sigma$ and left-multiplying by $\sigma$ corresponds to permuting rows according to $\sigma$.

Next, we define two types of symmetry that our query problem could exhibit.

\begin{defi}[Function symmetry]\label{defn:function_symmetry}
    A permutation $\sigma : D\rightarrow D$ is said to be \emph{function symmetric with respect to a function $F:D\rightarrow E$} if $F(\sigma(x)) = F(\sigma(y)) \iff F(x) = F(y)$, $\forall x,y\in D$. 
\end{defi}

\begin{defi}[Query symmetry]\label{defn:query_symmetry}
    A permutation $\sigma : D\rightarrow D$ is said to be \emph{query symmetric with respect to the set of queries $\cR$} if $q(\sigma(x)) = q(\sigma(y)) \iff q(x) = q(y)$, $\forall x,y \in D$.
\end{defi}
It is easy to observe that the above definition is equivalent to defining a permutation $\sigma : D\rightarrow D$ to be query symmetric with respect to $\cR$ if $\{\sigma\Delta_q\sigma^T \mid q\in \cR\} = \{ \Delta_q \mid q\in \cR \}$.

We now define a group that we use for the rest of this section and in the appendices.
\begin{defi}\label{def:bitflipgroup}
    Let $G$ denote the group generated by bit-flip permutations, i.e., $G = \{\sigma_y \mid y\in\{-1,1\}^n\}$ where $\sigma_y:x\mapsto x\oplus y$.
\end{defi}

As a first step, we prove that all elements in $G$ exhibit function symmetry with the parity function, and query symmetry with any substring-query set $\cQ$.

\begin{lemma}\label{lem:bit_flip_group_symmetry}
    Let $n\in\mathbb{N}$ and $\cal{Q}\subseteq 2^{[n]}$. Then, every $\sigma \in G$ is function-symmetric w.r.t.\ $\oplus$ and query-symmetric w.r.t.\ $\cQ$.
\end{lemma}

\begin{proof}

    Recall that $\forall x\in\{-1,1\}^n$, $|x\oplus y| \text{ is even } \iff \oplus(x) = \oplus(y)$. We use this fact to show that every $\sigma_z\in G$ is function symmetric with respect to the parity function.
    \begin{align*}
        & |\sigma_z(x)\oplus\sigma_z(y)| = |(x\oplus z)\oplus(y\oplus z)| = |x\oplus y|.
    \end{align*}
    Therefore,
    \begin{align*}
    \oplus(\sigma_z(x)) = \oplus(\sigma_z(y)) \iff |\sigma_z(x)\oplus\sigma_z(y)| \text{ is even } \iff |x\oplus y| \text{ is even } \iff \oplus(x) = \oplus(y).
    \end{align*}
    Let $S \in \cQ$, $b \in \pmone^S$, $\sigma_z$ be a permutation in $G$ and let $\sigma_{z\vert S}:x_S\mapsto x_S\oplus z_S$ denote the permutation restricted to indices in $S$. Then, 
    \begin{align*}
        \mathbb{I}[x_S = b] & = \mathbb{I}[(\sigma_z(x))_S = \sigma_{z\vert S}(b)].
    \end{align*}
    By the definition of our query set, the latter is a valid query on input $\sigma_z(x)$. Since this is true for all $S \in \cQ$, all $b \in \pmone^S$ and all $\sigma_z \in G$, this implies
    \begin{align*}
        \{\Delta_{S,b} \mid S\in\cal{Q},b\in\{-1,1\}^S\} & = \{\sigma_z\Delta_{S,b}\sigma_z^T \mid S\in\cal{Q},b\in\{-1,1\}^S\}.
    \end{align*}
    
    This shows that every $\sigma_z\in G$ is query symmetric with respect to $\cal{Q}$.\qedhere
\end{proof}

\begin{lemma}\label{lem:xor_structure_of_feasible_solution}
    Let $n$ be a positive integer and let $\cR$ be a query set. Let $G$ be the group generated by bit flip permutations. If every permutation $\sigma\in G$ is query symmetric  with respect to $\cR$, then there exists an $f : \pmone^n \to \bR$ such that 
    \begin{itemize}
        \item $f(x) = 0$ for all $x$ of even Hamming weight, and
        \item $M_{f \circ \XOR}$ is an optimal solution for the adversary bound $\ADV^{\cR, \pm}(\oplus)$.
    \end{itemize}
\end{lemma}

\begin{proof}
    Let $D=\pmone^n$. Let $(\beta, \Gamma)$ be an optimal solution to $\overline{\ADV}^{\cal{R}, \pm}(\oplus)$. In particular,
    \begin{equation*}
        \overline{\ADV}^{\cR,\pm}(\oplus) = \sum_{x,y\in\{-1,1\}^n}\Gamma[x,y].
    \end{equation*}

    For the next part of the proof, we show that $(\sigma\beta,\sigma\Gamma\sigma^T)$ is also a valid and optimal solution for all $\sigma \in G$. We organize this part of the proof by bullets, one for each constraint in the alternate formulation of the primal adversary bound.
    
    \begin{itemize}
        \item Entrywise non-negativity of $\sigma\beta$ is clear since this holds true for $\beta$, and hence any permutation of $\beta$ as well. Note that $\sigma\Gamma\sigma^T$ is simply reordering the rows and columns of $\Gamma$ by the same permutation $\sigma$. We observe below that $\sigma\Gamma\sigma^T$ is symmetric for any $\sigma\in G$.

    \begin{equation*}
        (\sigma\Gamma\sigma^T)^T = (\sigma^T)^T\Gamma^T\sigma^T = \sigma\Gamma^T\sigma^T = \sigma\Gamma\sigma^T.\tag{$\Gamma = \Gamma^T$}
    \end{equation*}
    \item Since every $\sigma \in G$ is function symmetric with respect to $\oplus$ (see \Cref{lem:bit_flip_group_symmetry}), we have $\oplus(x) = \oplus(y) \iff \oplus(\sigma(x)) = \oplus(\sigma(y))$. Therefore, $\sigma\Gamma\sigma^T[x,y] = \Gamma[\sigma(x),\sigma(y)] = 0$.
    \item Now we show that $\forall \sigma\in G, \quad \sigma\mathsf{diag}(\beta)\sigma^T -\sigma\Gamma\sigma^T\circ\Delta_q\succeq 0$. It is easy to see that every $\sigma\in G$ is its own inverse, and hence $\sigma^{-1} \in G$. Since every $\sigma \in G$ is query symmetric with respect to $\cR$, we have $\sigma^{-1}\Delta_q(\sigma^{-1})^T = \Delta_{q'}$ for some $q' \in \cR$ and hence $\Delta_q = \sigma\Delta_{q'}\sigma^T$. Let $\ket{v}\in\mathbb{R}^D$. 
    \begin{align*}
        \bra{v}\left(\sigma\mathsf{diag}(\beta)\sigma^T -\sigma\Gamma\sigma^T\circ\Delta_q \right)\ket{v} & = \bra{v}\left(\sigma\mathsf{diag}(\beta)\sigma^T -\sigma\Gamma\sigma^T\circ\sigma\Delta_{q^\prime}\sigma^T \right)\ket{v},\\
        & = \bra{v} \sigma(\mathsf{diag}(\beta) - \Gamma\circ\Delta_{q^\prime}) \sigma^T \ket{v},\\
        & = \bra{v^\prime}(\mathsf{diag}(\beta) - \Gamma\circ\Delta_{q^\prime})\ket{v^\prime},\tag{$\ket{v^\prime} = \sigma^T\ket{v}\in\mathbb{R}^D$}\\
        & \ge 0,
    \end{align*}
    where the last line uses $\beta - \Gamma\circ\Delta_q \succeq 0, \forall q\in \cR$, which holds true since $(\beta, \Gamma)$ is an optimal (and hence valid) solution to $\overline{\ADV}^{\cal{R}, \pm}(F)$.
    \item There are two cases to consider. From the function symmetry, we find that $\{\oplus^{-1}(1),\oplus^{-1}(-1)\} = \{\sigma(\oplus^{-1}(1)),\sigma(\oplus^{-1}(-1))\}$. That is, either $\sigma$ maps all $1$-inputs to $1$-inputs and $-1$-inputs to $-1$-inputs, or vice versa. In either case, we obtain that
    \[\sum_{x \in \oplus^{-1}(1)} (\sigma\beta)[x] = \sum_{x \in \oplus^{-1}(1)} \beta[\sigma^{-1}(x)] = \sum_{x \in \sigma(\oplus^{-1}(1))} \beta[x] = \frac12,\]
    and similarly when we take the summation over $x \in \oplus^{-1}(-1)$.
    \end{itemize}

    Thus, $(\sigma\beta, \sigma\Gamma\sigma^T)$ is a valid solution for all $\sigma \in G$. We now show its optimality. 
    Let $\ket{\bar{1}}$ denote the all-$1$ vector of length $2^n$. Then
    \begin{align*}
        \sum_{x,y\in \{-1,1\}^n}\Gamma[x,y] = \bra{\bar{1}}\Gamma\ket{\bar{1}} & = \bra{\bar{1}}\sigma^T\sigma\Gamma\sigma^T\sigma\ket{\bar{1}},\\
        = \bra{\bar{1}}\left(\sigma\Gamma\sigma^T\right)\ket{\bar{1}} & = \sum_{x,y\in\{-1,1\}^n}(\sigma\Gamma\sigma^T)[x,y].
    \end{align*}

    Next, we let $(\beta^\prime,\Gamma^\prime)$ be obtained by averaging over all permutations in $G$ as shown below.

    \begin{equation}
        (\beta^\prime,\Gamma^\prime) = \frac{1}{\left|G\right|}\sum_{\sigma\in G}(\sigma\beta,\sigma\Gamma\sigma^T).
    \end{equation}

    Since $\overline{\ADV}^{\cR,\pm}(\oplus)$ is a convex optimization program and $(\beta',\Gamma')$ is a linear combination of feasible solutions to $\overline{\ADV}^{\cR,\pm}(\oplus)$, we obtain that $(\beta',\Gamma')$ is also a feasible solution. Moreover, since for all $\sigma \in G$, all the objective values of the feasible solutions $(\sigma\beta,\sigma\Gamma\sigma^T)$ are equal as argued above, the objective value of $(\beta',\Gamma')$ is the same as that of $(\beta,\Gamma)$. We conclude that $(\beta',\Gamma')$ is also an optimal solution.

    Finally, we observe from \Cref{lem:equivalence_of_formulations} that an optimal solution to $\ADV^{\cR,\pm}(\oplus)$ can be written as
    \[M[x,y] = \frac{\Gamma'[x,y]}{\sqrt{\beta'[x]\beta'[y]}}.\]
    Since $\beta'$ and $\Gamma'$ are invariant under $\sigma$ by construction, we observe for all $\sigma \in G$ that
    \[M[\sigma(x),\sigma(y)] = \frac{\Gamma'[\sigma(x),\sigma(y)]}{\sqrt{\beta'[\sigma(x)]\beta'[\sigma(y)]}} = \frac{\Gamma'[x,y]}{\sqrt{\beta[x]\beta[y]}},\]
    and so in particular,
    \[
    M[x, y] = M[\sigma_y(x), \sigma_y(y)] = M[\sigma_y(x), 1^n] = M[x \oplus y, 1^n].
    \]
    Thus, if we define $f : \pmone^n \to \mathbb{R}$ as $f(x) = M[x,1^n]$, then we find that $M = M_{f \circ \mathsf{XOR}}$ is an optimal solution to $\ADV^{\cR,\pm}(\oplus)$. Finally, for any $x \in \pmone^n$, if $|x|$ is even, then $\oplus(x) = \oplus(1^n)$, and so $f(x) = M[x,1^n] = 0$.
\end{proof}

\begin{corollary}\label{cor:adv_xor_matrix}
    Let $n\in\mathbb{N}$ and $\cal{Q}\subseteq 2^{[n]}$. Then
    \begin{equation*}
        \ADV^{\cQ, \pm}(\oplus) = \max\limits_{\substack{f: \pmone^n \to \bR\\|z|~\textnormal{even}\implies f(z) = 0}} \frac{\|M_{f \circ \XOR}\|}{\max\limits_{\substack{S \in \cQ,\\ b \in \pmone^S}}\|M_{f \circ \XOR} \circ \Delta_{S,b}\|}.
    \end{equation*}
    
\end{corollary}

\begin{proof}
    We first show that $M_{f\circ\XOR}$ is a feasible solution to $\ADV^{\cQ,\pm}(\oplus)$, for all $f:\{-1,1\}^n\rightarrow\mathbb{R}$ satisfying $f(x) = 0$ when $|x|$ is even. The matrix $M_{f\circ\XOR}$ is symmetric since $M[x, y] = f(x\oplus y) = f(y\oplus x) = M[y,x]$ for all $x,y\in\{-1,1\}^n$. We know that $|x\oplus y|$ is even if both $|x|, |y|$ are even or if both $|x|, |y|$ are odd. Therefore,
    \begin{equation*}
        \oplus(x) = \oplus(y) \implies |x\oplus y| \text{ is even } \implies \ M_{f\circ\XOR}[x,y] = 0.
    \end{equation*}
    \Cref{lem:xor_structure_of_feasible_solution} (which is applicable since \Cref{lem:bit_flip_group_symmetry} guarantees all elements of $G$ to be query symmetric w.r.t.~$\cQ$ and function symmetric w.r.t.~$\oplus$) yields the corollary.
\end{proof}

\subsection{Simplification by Fourier Analysis}

In this subsection we show that using the structure we derived in the last subsection, we can simplify the expression on the right-hand side of \Cref{cor:adv_xor_matrix} to a concise analytic optimization problem over the Boolean hypercube. Lemma~\ref{lem:xornorm} says that the numerator equals $2^n \max_{S \subseteq [n]}|\widehat{f}(S)|$. For the denominator, we require the following lemma.

\begin{lemma}\label{lem:blockmatrix}
    For all $f : \pmone^n \to \bR$, and all $S \subseteq [n], b \in \pmone^S$, we have
    \begin{equation}
        \setlength\aboverulesep{0pt}
        \setlength\belowrulesep{0pt}
        \setlength\cmidrulewidth{0.5pt}
        \label{eq:matrix}
        M_{f \circ \XOR} \circ \Delta_{S, b} = \begin{blockarray}{rc|ccc}
            & y_S = b && y_S \neq b & \\
            \begin{block}{r[c|ccc]}
            x_S = b & 0 & \cdots & M_{b, y_S} & \cdots \\\cmidrule(lr){1-5}
            & \vdots &&& \\
            x_S \neq b & M_{x_S, b} && 0 & \\
            & \vdots &&& \\
            \end{block}
        \end{blockarray},
    \end{equation}
    where for all $x_S, y_S$, $M_{x_S, y_S} := M_{f_{\bar{S}|x_S \oplus y_S} \circ \mathsf{XOR}}$.
\end{lemma}
\begin{proof}
    We verify the equality entry-wise. To that end, let $x,y \in \pmone^n$. If $\Delta_{S,b}[x,y] = 0$, then the left-hand side is $0$. Moreover, we have $\bI[x_S = b] = \bI[y_S = b]$, i.e., either $x_S = b$ and $y_S = b$, or $x_S \neq b$ and $y_S \neq b$. In both cases, we also have a $0$-entry on the right-hand side as well.

    On the other hand, suppose that $\Delta_{S,b}[x,y] = 1$. Then, we know that either $x_S = b$ or $y_S = b$, but not both. Then, on the left-hand side of the equation, we have $M_{f \circ \mathsf{XOR}}[x,y] = f(x \oplus y)$, and on the right-hand side we have
    \begin{equation*}
        M_{f_{\bar{S}|x_S \oplus y_S} \circ \mathsf{XOR}}[x_{\bar{S}},y_{\bar{S}}] = f_{\bar{S}|x_S \oplus y_S}(x_{\bar{S}} \oplus y_{\bar{S}}) = f(1_Sx_{\bar{S}} \oplus 1_Sy_{\bar{S}} \oplus 1_{\bar{S}}x_S \oplus 1_{\bar{S}}y_S) = f(x \oplus y).
    \end{equation*}
    In the above equation, $1_Sx_{\bar{S}}$ denotes the string that equals $1_S$ on $S$ and $x$ on $\bar{S}$. The strings $1_Sy_{\bar{S}}$, $1_{\bar{S}}x_S$ and $1_{\bar{S}}y_S$ are defined analogously.\qedhere
\end{proof}

Now, we turn to the computation of the spectral norm of this matrix. Recall the definition of $D_{\widehat{f}}$ from \Cref{fact:xorhadamard}.

\begin{lemma}\label{lem:diagonal-blocks}
    For all $f : \pmone^n \to \mathbb{R}$, and all $S \subseteq [n]$ and $b \in \pmone^S$, we have
    \[\norm{M_{f \circ \mathsf{XOR}} \circ \Delta_{S,b}} = \setlength\aboverulesep{0pt}
    \setlength\belowrulesep{0pt}
    \setlength\cmidrulewidth{0.5pt}
    \norm{\begin{blockarray}{rc|ccc}
    & y_S = b && y_S \neq b & \\
    \begin{block}{r[c|ccc]}
    x_S = b & 0 & \cdots & D_{b,y_S} & \cdots \\\cmidrule(lr){1-5}
    & \vdots &&& \\
    x_S \neq b & D_{x_S,b} && 0 & \\
    & \vdots &&& \\
    \end{block}
\end{blockarray}},\]
where $D_{x_S,y_S} := D_{2^{n-|S|}\widehat{f_{\bar{S}|x_S \oplus y_S}}}$. The rows and columns of the matrix on the RHS are indexed by $x,y \in \pmone^n$.
\end{lemma}

\begin{proof}
    It suffices to show that the matrices on the left- and right-hand side are related by a similarity transform, since they leave the spectrum, and hence spectral norm, invariant. To that end, let $H_{\bar{S}}$ be the Hadamard transform over $\bar{S}$, and let $X = H_{\bar{S}} \otimes I_S$, where $I_S \in \mathbb{R}^{\pmone^S \times \pmone^S}$ is the identity matrix. We observe that $X^{-1} = H_{\bar{S}}^{-1} \otimes I_S$, and so combining \Cref{lem:blockmatrix}~and~\Cref{fact:xorhadamard}, we observe that
    \[\setlength\aboverulesep{0pt}
    \setlength\belowrulesep{0pt}
    \setlength\cmidrulewidth{0.5pt}
    X^{-1}(M_{f \circ \mathsf{XOR}} \circ \Delta_{S,b})X = \begin{blockarray}{rc|ccc}
        & y_S = b && y_S \neq b & \\
        \begin{block}{r[c|ccc]}
        x_S = b & 0 & \cdots & H_{\bar{S}}^{-1}H_{\bar{S}}D_{b, y_S}H^{-1}_{\bar{S}}H_{\bar{S}} & \cdots \\\cmidrule(lr){1-5}
        & \vdots &&& \\
        x_S \neq b & H_{\bar{S}}^{-1}H_{\bar{S}}D_{x_S, b}H^{-1}_{\bar{S}}H_{\bar{S}} && 0 & \\
        & \vdots &&& \\
        \end{block}
    \end{blockarray},\]
    which equals the matrix on the RHS of the equation in the lemma statement.
\end{proof}

The following lemma rewrites the spectral norm further. Specifically, using the structure of the matrix in the RHS of \Cref{lem:diagonal-blocks} we are able to rewrite the matrix as a direct sum of $2^{|\bar{S}|}$ many matrices, each of which is $2^{|S|}$-dimensional, with rows and columns indexed by strings in $\pmone^S$.

\begin{lemma}\label{lem:diagonal-blpcks}
    For all $f : \pmone^n \to \mathbb{R}$, and all $S \subseteq [n]$ and $b \in \pmone^S$, we have
    \[\norm{M_{f \circ \mathsf{XOR}} \circ \Delta_{S,b}} = 2^{n-|S|}\max_{z \in \pmone^{\bar{S}}} \setlength\aboverulesep{0pt}
    \setlength\belowrulesep{0pt}
    \setlength\cmidrulewidth{0.5pt}
    \norm{\begin{blockarray}{rc|ccc}
    & y = b && y \neq b & \\
    \begin{block}{r[c|ccc]}
    x = b & 0 & \cdots & \widehat{f_{\bar{S}|b \oplus y}}(z) & \cdots \\\cmidrule(lr){1-5}
    & \vdots &&& \\
    x \neq b & \widehat{f_{\bar{S}|x \oplus b}}(z) && 0 & \\
    & \vdots &&& \\
    \end{block}
\end{blockarray}},\]
where the row and column labels range over $x,y \in \pmone^S$.
\end{lemma}

\begin{proof}
    We start from the matrix on the RHS in \Cref{lem:diagonal-blocks}. We observe that its $[x,y]$'th entry is only non-zero when $x_{\bar{S}} = y_{\bar{S}}$, since each of the blocks $D_{x_S,y_S}$ is diagonal. Thus, by rearranging rows and columns of this matrix by the same permutations (which doesn't affect spectral norm), we obtain a block-diagonal matrix, with every block labeled by $z \in \pmone^{\bar{S}}$. Hence, its operator norm is the maximum operator norm of each of these blocks. The block corresponding to $z \in \pmone^{\bar{S}}$ is precisely the matrix that appears in the RHS of the lemma statement, which concludes the proof.
\end{proof}

Now, we can prove the required bound on the operator norm of the denominator in \Cref{cor:adv_xor_matrix}.

\begin{theorem}\label{thm:denominatornorm}
    For all $f : \pmone^n \to \bR, S \subseteq [n]$ and all $b \in \pmone^S$, we have
    \[
    \|M_{f \circ \mathsf{XOR}} \circ \Delta_{S,b}\| = 2^n \max_{x \in \pmone^{\bar{S}}} \sqrt{\Var(\hat{f}_{S|x})}.
    \]
\end{theorem}
\begin{proof}
For $x \in \pmone^S$ and $y \in \pmone^{\bar{S}}$, let $xy$ denote the bit string that equals $x$ on $S$ and $y$ on $\bar{S}$. For the rest of this proof we identify a bitstring with its characteristic set. Specifically, this means for every $f : \pmone^n \to \bR$, we use the notation $\hat{f}(x) := \hat{f}(\{j \in [n] : x_j = -1\})$.
The spectral norm of a symmetric matrix with a single non-zero row is simply the norm of that row. We can thus combine \Cref{lem:diagonal-blpcks} and \Cref{fact:off-diagonal-matrix-norm} to obtain that
\begin{align*}
    \|M_{f \circ \mathsf{XOR}} \circ \Delta_{S,b}\| &= 2^{n-|S|} \max_{x \in \pmone^{\bar{S}}} \norm{\left(\widehat{f_{\bar{S}|y \oplus b}}(x)\right)_{y \in \pmone^S \setminus \cbra{b}}}_2 \\
    &= 2^{n-|S|} \max_{x \in \pmone^{\bar{S}}} \sqrt{\sum_{y \in \pmone^S  \setminus \{b\}} \widehat{f_{\bar{S}|y \oplus b}}(x)^2} \\
    &= 2^{n-|S|} \max_{x \in \pmone^{\bar{S}}} \sqrt{\sum_{y \in \pmone^S  \setminus \{1^S\}} \widehat{f_{\bar{S}|y}}(x)^2} \tag*{by variable substitution}\\
    & = 2^{n-|S|} \max_{x \in \pmone^{\bar{S}}} \sqrt{\sum_{y \in \pmone^S}\widehat{f_{\bar{S}|y}}(x)^2 - \widehat{f_{\bar{S}|1^S}}(x)^2}\\
    & = 2^{n-|S|} \max_{x \in \pmone^{\bar{S}}} \sqrt{2^{|S|}\E_{y \in \pmone^S}[\widehat{f_{\bar{S}|y}}(x)^2] - \widehat{f_{\bar{S}|1^S}}(x)^2}\\
    & = 2^{n-|S|} \max_{x \in \pmone^{\bar{S}}} \sqrt{2^{|S|}\sum_{y \in \pmone^{S}}\widehat{f}(xy)^2 - \left(\sum_{y \in \pmone^S}\widehat{f}(xy)\right)^2}\tag*{By Corollaries~\ref{cor:0restcoeffs} and~\ref{cor:parsevalrestcoeffs}}\\
    & = 2^{n-|S|} \max_{x \in \pmone^{\bar{S}}} \sqrt{2^{2|S|}\E[\hat{f}_{S|x}^2] - 2^{2|S|}\E[\hat{f}_{S|x}]^2} \\
    & = 2^n \max_{x \in \pmone^{\bar{S}}} \sqrt{\Var(\hat{f}_{S|x})}.
\end{align*}
\end{proof}

Observe that the RHS in the theorem above is independent of $b$. Using this derived expression for the denominator, we are able to now state and prove our main theorem of this section. Recall the definition of $\val_\cQ$ from \Cref{eq:dfn-valq}.

\begin{repeattheorem}{\Cref{thm:main}}
    Let $n$ be a positive integer and $\cQ \subseteq 2^{[n]}$. 
    Then $\Q^\cQ(\oplus) = \Theta(\val_\cQ)$.
\end{repeattheorem}

\begin{proof}
    By \Cref{cor:adv_xor_matrix}, we have
    \[
        \ADV^{\cQ, \pm}(\oplus) = \max\limits_{\substack{f: \pmone^n \to \bR,\\|z|~\textnormal{even}\implies f(z) = 0}} \frac{\|M_{f \circ \XOR}\|}{\max\limits_{\substack{S \in \cQ,\\ b \in \pmone^S}}\|M_{f \circ \XOR} \circ \Delta_{S,b}\|}.
    \]
    We can use Lemma~\ref{lem:xornorm} and Theorem~\ref{thm:denominatornorm} to rewrite the numerator and denominator, respectively. This implies that
    \[
        \ADV^{\cQ, \pm}(\oplus) = \max\limits_{\substack{f: \pmone^n \to \bR,\\|z|~\textnormal{even}\implies f(z) = 0}} \frac{2^n \|\hat{f}\|_\infty}{2^n \max\limits_{S \in \cQ} \max\limits_{b \in \pmone^{\bar{S}}} \sqrt{\Var(\widehat{f}_{S|b})}}.
    \]
    Lemma~\ref{lem:fourierequiv} says that the condition `$|z|$ even $\implies f(z) = 0$' is equivalent to the condition `for all $S \subseteq [n]$, we have $\widehat{f}(S) = - \widehat{f}(\bar{S})$' (equivalently, $\widehat{f}$ is an odd function, identifying sets with their characteristic vectors). 
    The RHS above is thus equal to
    \[
        \max\limits_{\substack{\widehat{f}: 2^{[n]} \to \bR,\\\widehat{f} \textnormal{ is an odd function}}} \frac{2^n \|\hat{f}\|_\infty}{2^n \max\limits_{S \in \cQ} \max\limits_{b \in \pmone^{\bar{S}}} \sqrt{\Var(\widehat{f}_{S|b})}}.
    \]
    Syntactically replacing $\widehat{f}$ by $f$ everywhere and using \Cref{thm:adv_equals_quantum_queries}, we obtain
    \[
    \Q^{\cQ}(\oplus) = \Theta(\ADV^{\cQ, \pm}(\oplus)) = \Theta\left(\max_{\substack{f : \pmone^n \to \bR,\\ f \textnormal{ is an odd function}}} \frac{\|f\|_\infty}{\max\limits_{\substack{S \in \cQ,\\ b \in \pmone^{\bar{S}}}}\sqrt{\Var(f_{S|b})}}\right).\qedhere
    \]
\end{proof}

Using Observation~\ref{obs:paritylowerbound} and Lemma~\ref{lem:downwardclosed}, we obtain the following immediate corollary about the quantum $\cQ$-substring quantum query complexity of learning an input string.
\begin{corollary}\label{cor:learning}
    Let $n \in \mathbb{N}$ and $\cQ \subseteq 2^{[n]}$.
    Then, $\Q^\cQ(\rec) = \Omega(\val_\cQ)$. If $\cQ$ is also downward closed, then $\Q^\cQ(\rec) = \Theta(\val_\cQ)$.
\end{corollary}

\section{Applications of the Framework}\label{sec:apps}
In this section we use the framework developed in the last section to analyze $\val_\cQ$, and thereby the quantum $\cQ$-substring query complexity of computing parity, for various interesting collections $\cQ$.

\subsection{Notation}
We define some notation that we will use for the rest of this section.

For a collection of sets $\cQ \subseteq 2^{[n]}$, we use the notation $\val_{\cQ}$ as defined in Theorem~\ref{thm:main}.
That is,
\begin{tcolorbox}
\begin{equation}\label{eq:valq}
    \val_\cQ := \max_{\substack{f : \pmone^n \to \bR\\ f \textnormal{ is an odd function}}} \frac{\|f\|_\infty}{\max\limits_{\substack{S \in \cQ,\\ b \in \pmone^{\bar{S}}}}\sqrt{\Var(f_{S|b})}}.
\end{equation}
\end{tcolorbox}

Recall the following interpretation of $\val_\cQ$: it is the maximum, over all odd functions $f : \pmone^n \to \bR$, ratio of the maximum absolute value of $f$ to the maximum standard deviation of $f$ in a subcube. The subcubes under consideration here are precisely those where the set of free variables must be equal to an allowed substring query set (i.e., a set in $\cQ$). The fixed variables can be fixed to any value. We organize the rest of this section into various subsections, one for each type of $\cQ$ that we consider.

\subsection{Contains All Singletons}
In the case that the allowed substring queries include all singletons, there is an easy quantum query upper bound of $n$: simply query all of the variables one at a time. As a sanity check, we first show how this can be recovered by our framework.

\begin{proposition}\label{prop:singletons}
    Let $\cQ \supseteq \cbra{\cbra{i} : i \in [n]}$. Then $\val_\cQ \leq n$.
\end{proposition}
\begin{proof}
Since $f$ is odd, there is always an $x^*$ such that $f(x^*)=\|f\|_\infty>0$. Let $M = \max_{x}f(x)$ be attained at $x^*$. As $f$ is odd, $f(-x^*) = -M$.

Consider a shortest path from $x^*$ to $-x^*$ in the hypercube, this has length $n$. Along this path the total value of $f$ decreases by $2M$, thus there must exist an edge with a drop of at least $2M/n$. That edge corresponds to flipping one coordinate, say $i \in [n]$. This means there exists an $z \in \pmone^n$ such that $|f(z) - f(z \oplus e_i)| \geq 2M/n$. Choose the subcube that fixes all variables but $i$ to be consistent with $z$, and leave the $i$'th variable free. Note that this is a valid subcube to consider as $\cQ$ contains all singletons. For this choice of subcube we have the variance (using the formula $\Var(g) = \frac{1}{2}\E[(g(x) - g(y))^2]$) to be $\frac{(f(z) - f(z \oplus e_i))^2}{4} \geq \frac{4M^2}{4n^2}$. Thus $\val_\cQ \leq \frac{M}{\sqrt{4M^2/4n^2}} = n$.
\end{proof}

\subsection{Sets of Bounded Size}
The setting $\cQ = \cbra{S \subseteq 2^{[n]} : |S| \leq k}$ corresponds to the model where the allowed substrings to be queried are bounded in size.

In particular, $k = n$ corresponds to the usual search with wildcards setting, where an algorithm is allowed to query OR's of arbitrary sets of literals. 

\begin{theorem}\label{thm:valbounded}
    Let $1 \leq k \leq n$ be an integer, and let $\cQ = \cbra{S \subseteq 2^{[n]} : |S| \leq k}$. Then,
    \[
    \frac{n}{\sqrt{k}} \leq \val_\cQ \leq \frac{2n}{\sqrt{k}}.
    \]
\end{theorem}
We show the lower bound using the (appropriately shifted) Hamming weight function. We prove the upper bound by induction on $n$. We first state a convenient induction hypothesis, then show how it implies the upper bound, and finally prove the induction hypothesis.

\begin{lemma}\label{lem:induction}
    For all integers $n > 0$ and all functions $f : \pmone^n \to \bR$ with $\E[f(x)] = 0$, there exists a $S \subseteq [n]$ and $b \in \pmone^{\bar{S}}$ where:
    \begin{itemize}
        \item $\Var(f_{S|b}) \geq \frac{\|f\|_\infty^2}{4n}$, and
        \item $\left|\widehat{f_{S|b}}(\cbra{i})\right| \geq \frac{\|f\|_\infty}{2n}$ for all $i \in S$.
    \end{itemize}
\end{lemma}

In other words, there exists a subcube where the restricted function has large variance, as well as large degree-1 Fourier coefficients. We feel this is an independently interesting result in Fourier analysis.
We first show how Theorem~\ref{thm:valbounded} follows from Lemma~\ref{lem:induction} (note that Lemma~\ref{lem:induction} gives no guarantee on $|S|$, which we require for proving Theorem~\ref{thm:valbounded}).

\begin{proof}[Proof of Theorem~\ref{thm:valbounded}]
We first show the lower bound on $\val_\cQ$, and then the upper bound.
\begin{itemize}
    \item \textbf{Lower bound:} Consider the function $f : \pmone^n \to \bR$ defined by $f(x) = |x| - (n/2)$. It is clear that $f$ is an odd function since $f(-x) = |-x| - (n/2) = (n - |x|) - (n/2) = n/2 - |x| = -f(x)$. Let $m \leq k$. For any subcube that fixes all but a collection of $m$ variables, it is easy to see that the random variable obtained by restricting $f$ to this subcube and picking a uniformly random function value in the restricted subcube is simply a constant additive shift of the binomial random variable $\mathrm{Bin}(m, 1/2)$. Since such shifts do not affect variance, the variance in this subcube equals $m/4 \leq k/4$, since the only subcubes under consideration are those where the number of free variables is at most $k$. Since $\max |f(x)| = n/2$, this function $f$ witnesses
    \[
    \val_\cQ \geq \frac{n/2}{\sqrt{k/4}} = \frac{n}{\sqrt{k}}.
    \]
    \item \textbf{Upper bound:} Let $M := \|f\|_\infty = \max_{x \in \pmone^n}|f(x)|$. Let $S \subseteq [n]$ and $b \in \pmone^{\bar{S}}$ be obtained from Lemma~\ref{lem:induction} (which is applicable since $f$ is odd, and hence balanced).
    
    If $|S| \leq k$, then this is a valid subcube that can be chosen in the denominator of $\val_\cQ$, and we get $\val_\cQ \leq \frac{M}{\sqrt{M^2/4n}} = 2\sqrt{n} \leq \frac{2n}{\sqrt{k}}$, which proves the theorem. So for the rest of the proof we may assume that $|S| > k$. For convenience, define $g:=f_{S|b} : \pmone^S \to \bR$.
    
    Let $T$ be any subset of $S$ of size $k$, and define $\bar{T}:=S \setminus T$ (i.e., the complement of $T$ with respect to the universe $S$). For an arbitrary $x \in \pmone^{\bar{T}}$, consider the function $g_{T|x} : \pmone^T \to \bR$ obtained by restricting $g$ further by setting the variables in $\bar{T}$ to $x$. We have for each $i \in T \subseteq S$,
    \begin{align*}
        \E_{x}[\widehat{g_{T|x}}(\cbra{i})^2] & = \sum_{A \subseteq \bar{T}} \widehat{g}(A \cup \cbra{i})^2 \tag*{by Corollary~\ref{cor:parsevalrestcoeffs}}\\
        & \geq \widehat{g}(\emptyset \cup \cbra{i})^2 = \widehat{g}(\cbra{i})^2\\
        & \geq \frac{M^2}{4n^2},
    \end{align*}
    where the last line holds by Lemma~\ref{lem:induction}.
    By linearity of expectation, and since $|T| = k$, we have
    \begin{align*}
        \E_x\left[\sum_{i \in T}\widehat{g_{T|x}}(\cbra{i})^2\right] = \sum_{i \in T}\E_x\left[\widehat{g_{T|x}}(\cbra{i})^2\right] \geq \frac{M^2k}{4n^2}.
    \end{align*}
    Therefore there exists a $x^* \in \pmone^{\bar{T}}$ such that $\sum_{i \in T}\rbra{\widehat{g_{T|x^*}}(\cbra{i})^2} \geq \frac{M^2 k}{4n^2}$. Since $\Var(h) = \sum_{S \neq \emptyset} \widehat{h}(S)^2$ for all functions $h$ by Lemma~\ref{lem:fouriervariance}, we have
    \begin{align*}
        \Var(g_{T|x^*}) = \sum_{S \subseteq T, S \neq \emptyset}\widehat{g_{T|x^*}}(S)^2 \geq \sum_{i \in T}\widehat{g_{T|x^*}}(\cbra{i})^2 \geq \frac{M^2 k}{4n^2}.
    \end{align*}
    Note that $g_{T|x^*}$ is the restriction of $f$ to the subcube $T$, by setting the variables in $\bar{S}$ and $\bar{T}$ to $b$ and $x^*$, respectively. This is a valid subcube for the denominator of $\val_\cQ$ (see \Cref{eq:valq}), and so this subcube witnesses
    \[
    \val_\cQ \leq \frac{M}{\sqrt{\frac{M^2k}{4n^2}}} = \frac{2n}{\sqrt{k}},
    \]
    which completes the proof.\qedhere
\end{itemize}
\end{proof}

It only remains to prove our induction hypothesis, Lemma~\ref{lem:induction}, which we do now.

\begin{proof}[Proof of \Cref{lem:induction}]
We prove the lemma by induction on $n$. Let $\|f\|_\infty=M$ and assume that $|f(x^*)|=M$.
    \begin{itemize}
        \item \textbf{Base case}: $n = 1$. Thus we have $f(-1) = M$ and $f(1) = -M$ (or the opposite, in which case essentially the same proof works). We have the variance over the whole cube as $\Var(f) = M^2 > \frac{M^2}{4}$. We also have $\widehat{f}(\cbra{1}) = $ either $M$ or $-M$. In either case, $|\widehat{f}(\cbra{1})| = M \geq \frac{M}{2}$. This proves the base case.
        \item \textbf{Assumption for inductive step}: We assume the induction hypothesis to be true for $n = \ell$, i.e., for every balanced function $f$ on $\ell$ variables, there exists a $S \subseteq [\ell]$ and $b \in \pmone^{\bar{S}}$ such that $\Var(f_{S|b}) \geq \frac{M^2}{4\ell}$, and $\left|\widehat{f_{S|b}}(\cbra{i})\right| \geq \frac{M}{2\ell}$ for all $i \in S$.
        \item \textbf{Inductive step}: For the inductive step, we assume that the hypothesis holds for $n=\ell$ (see above) and show it is true for $n=\ell + 1$. Define $\mu_{i, b} := \E_{x : x_i = b}[f(x)]$ for $b \in \pmone$. We consider two cases.
        \begin{itemize}
            \item \textbf{Case 1:} $\exists i \in [\ell + 1]$ such that $|\mu_{i, x^*_i} - \mu_{i, -x^*_i}| \leq \frac{M}{\ell + 1}$. Since we know $f$ to be balanced, we have $\mu_{i, x^*_i} + \mu_{i, -x^*_i} = 0$, and hence by the triangle inequality
            \[
                |\mu_{i, x^*_i}| = \left|\frac{\mu_{i,x^*_i} - \mu_{i,-x^*_i}}{2} + \frac{\mu_{i,x^*_i} + \mu_{i,-x^*_i}}{2}\right| \leq \frac{|\mu_{i,x^*_i} - \mu_{i,-x^*_i}|}{2} + \frac{|\mu_{i,x^*_i} + \mu_{i,-x^*_i}|}{2} \leq \frac{M}{2(\ell + 1)}.
            \]
            Define $f_i := f|_{x_i = x^*_i} - \mu_{i, x^*_i}$. Note that we have $\E[f_i] = 0$ and $\|f_i\|_\infty \geq M - \frac{M}{2(\ell + 1)}$, where we used that the maximum absolute value of $f$ is attained in the domain of $f_i$ by construction (i.e., $x^*$ is in the domain of $f_i$). Applying the induction hypothesis to $f_i$, we get a $S \subseteq [\ell + 1]$ and $b \in \pmone^{\bar{S}}$ (with $i \in \bar{S}$ and $b_i=x^*_i$) such that the following two properties hold.
            \begin{itemize}
                \item We write $b' := b|_{\bar{S} \setminus \{i\}}$, i.e., the bit string $b$ without the bit at position $i$, and obtain
                \[
                \Var((f_i)_{S|b'}) \geq \frac{\rbra{M\rbra{1 - \frac{1}{2(\ell + 1)}}}^2}{4\ell} \geq \frac{\rbra{M\rbra{\sqrt{1 - \frac{1}{\ell + 1}}}}^2}{4\ell} = \frac{M^2}{4\ell}\cdot\frac{\ell}{\ell + 1} = \frac{M^2}{4(\ell + 1)}.
                \]
                Here we are using $1-\frac{1}{2\alpha} \geq \sqrt{1 - \frac{1}{\alpha}}$ for all $\alpha > 1$, which can easily be verified by squaring both sides and comparing terms.
                Since $f_{S|b} = (f_i)_{S|b'} + \mu_{i,x_i^*}$ and variance is unaffected by a constant additive shift, this implies $\Var(f_{S|b}) \geq \frac{M^2}{4(\ell + 1)}$ as well.
                \item First note that for any function $g : \pmone^n \to \bR$ and any constant $c$, all Fourier coefficients of $g$ are equal to the corresponding Fourier coefficients of $g - c$, except for the empty coefficient which gets subtracted by $c$. Thus, we again write $b' := b|_{\bar{S} \setminus \{i\}}$, to find that for all $j \in S$, 
                \[
                    |\widehat{f_{S|b}}(\cbra{j})| = |\widehat{(f_i)_{S|b'}}(\cbra{j})|  \geq \frac{M\rbra{1 - \frac{1}{2(\ell + 1)}}}{2\ell} = \frac{M \cdot \frac{2\ell + 1}{2(\ell + 1)}}{2\ell} > \frac{M \cdot 2\ell}{2(\ell + 1)}\cdot\frac{1}{2\ell} = \frac{M}{2(\ell + 1)}.
                \]
            \end{itemize}
                \item \textbf{Case 2:} For all $i \in [\ell + 1]$, we have $|\mu_{i, x^*_i} - \mu_{i,-x^*_i}| > \frac{M}{\ell + 1}$. In that case, we choose $S = [\ell+1]$ and $b$ becomes the empty string, so that $f_{S|b} = f$. Then, we find
                \begin{itemize}
                    \item By definition of Fourier coefficients, we have for all $i \in [\ell+1]$,
                    \[
                        | \widehat{f}(\cbra{i})| = | \E_x[f(x)\cdot x_i] | = \left|\frac{x^*_i}{2} (\mu_{i, x^*_i} - \mu_{i, -x^*_i})\right| > \frac{M}{2(\ell + 1)}.
                    \]
                    \item By Lemma~\ref{lem:fouriervariance}, we have the variance of the whole subcube as
                    \[
                    \Var(f) = \sum_{S \neq \emptyset}\widehat{f}(S)^2 \geq \sum_{i \in [\ell + 1]}\widehat{f}(\cbra{i})^2 > (\ell + 1)\frac{M^2}{4(\ell + 1)^2} = \frac{M^2}{4(\ell + 1)}.\qedhere
                    \]
            \end{itemize}
        \end{itemize}
    \end{itemize}
\end{proof}

\subsection{Contiguous Blocks}

The setting $\cQ = \cbra{[i, j] : i \leq j \in [n]}$ corresponds to the model where the allowed substrings to be queried form contiguous blocks. While we don't phrase it as such, our results in this subsection also hold when the notion of contiguity allows for wraparound.

The following is our main result of this subsection.
\begin{theorem}\label{thm:valcontig}
    Let $n$ be a positive integer and let $\cQ = \cbra{[i, j] : i \leq j \in [n]}$. Then there exists a universal constant $c$ such that
    \[
        \frac{cn}{\log n} \leq \val_\cQ \leq n.
    \]
\end{theorem}
\begin{proof}
Since $\cQ$ contains all singletons, Proposition~\ref{prop:singletons} implies the upper bound.
For the lower bound, define $f : \pmone^{mk} \to \bR$ as $f(z^{(1)}, \dots, z^{(m)}) = \sum_{i = 1}^m g(z^{(i)})$. Here each $z^{(i)} \in \pmone^k$, and 
\[
g(z^{(i)}) = \begin{cases}
    1, & z^{(i)} = 1^k, \\
    -1, & z^{(i)} = (-1)^k, \\
    0, & \textnormal{otherwise}.
\end{cases}
\]
First note that $f$ is odd and $\max|f(x)| = m = n/k$. We show that every valid subcube in the denominator of the RHS of the expression of $\val_\cQ$ (see Equation~\eqref{eq:valq}) has small variance, which proves the theorem.

For any subcube $(S, b)$ where $S$ is a contiguous block, since a contiguous block can partially intersect at most 2 blocks (let the free variables in these blocks be denoted by strings $y^{(1)}$ and $y^{(2)}$), we can assume that the restriction of $f$ to the subcube $(S, b)$ has the following form:
\begin{equation}\label{eq:contigsubcube}
f_{S| b}(y^{(1)}, y^{(2)}, z^{(1)}, z^{(2)}, \dots z^{(t)}) = h_1(y^{(1)}) + h_2(y^{(2)}) + \sum_{i = 1}^t g(z^{(i)}) + c,
\end{equation}
where $h_1, h_2 \in \cbra{r_1, r_{-1}, r_0}$ depending on the values of the set variables in the corresponding blocks, where for all $b \in \cbra{-1, 0, 1}$,
\[
r_b(y) := \begin{cases}
    b, & \textnormal{if } y \textnormal{ is the all-}b \textnormal{ string}, \\
    0, & \textnormal{otherwise}.
    \end{cases}
\]
In particular, $r_0$ is the constant 0 function. Since each block in the restricted function contains disjoint variables, the variance of the sum on the RHS of Equation~\eqref{eq:contigsubcube} equals the sum of variances of the individual terms, which equals
\begin{align*}
& = \Var(h_1) + \Var(h_2) + \sum_{i = 1}^t\Var(g) + 0,\\
& = O(1) + t \cdot \E[g^2],\\
& = O(1) + m \cdot \frac{1}{2^{k-1}},
\end{align*}
where the second line follows since since $h_i$ is a $\cbra{-1, 0, 1}$-valued function, and each $g$ is a balanced function, and the last line follows since $t \leq m$ and $g^2$ takes value 1 at precisely two points and 0 everywhere else.

If we set $n$ to be the number of variables for $f$, which equals $mk$, then the above variance equals $O(1) + \frac{n/k}{2^{k - 1}}$. Thus we have 
\[
\val_\cQ = \frac{\|f\|_\infty}{\sqrt{\max\limits_{\substack{S \in \cQ,\\ b \in \pmone^{\bar{S}}}}\Var({f_{S | b}})}} = \frac{n/k}{\sqrt{O(1) + \frac{n/k}{2^{k - 1}}}} = \Omega\left(\min\cbra{\frac{n}{k}, \sqrt{\frac{n}{k}}\cdot2^{\frac{k - 1}{2}}}\right).
\]
Setting $k$ to maximize this quantity asymptotically, we get $\frac{n}{k} = 2^{k-1}$, which implies $n = k\cdot 2^{k-1}$, implying $k = \log n - \log\log n + O(1)$, and hence $\val_\cQ = \Omega(n/\log n)$, proving the theorem.
\end{proof}

\subsection{Prefixes}

In this subsection we deal with the setting where $\cQ = \cbra{[1, i] : i \in [n]}$, that is, the allowed substrings to be queried are only prefixes of the input string. 
\begin{theorem}\label{thm:valprefix}
    Let $n$ be a positive integer and let $\cQ = \cbra{[1,i] : i \in [n]}$. Then,
    \[
        \frac{n}{\sqrt{8}} \leq \val_\cQ \leq n.
    \]
\end{theorem}

We first define decision lists.
\begin{defi}[Decision lists]\label{def:declist}
A decision list is a sequence $D = ((L_1, a_1), (L_2, a_2), \dots, (L_k, a_k), b)$, where each $a_i, b \in \bR$ and each $L_i$ is either a variable or the negation (i.e. complement) of a variable.  The decision list computes a function $f : \cbra{-1, 1}^n \rightarrow \bR$ as follows.
If $L_1(x) = -1$, then $f(x) = a_1$; elseif $L_2(x) = -1$, then $f(x) = a_2$, elseif \dots, elseif $L_k(x) = -1$, then $f(x) = a_k$, else $f(x) = b$.
\end{defi}
\begin{proof}[Proof of Theorem~\ref{thm:valprefix}]
We first show the upper bound and then the lower bound.
\begin{itemize}
    \item The upper bound follows immediately from the observation that the values of all the variables (and hence their parity) can be learned by $n$ (classical) prefix queries and the characterization given in Theorem~\ref{thm:main}. Below we recover the same bound by analyzing $\val_\cQ$ directly. To this end, let $f:\pmone^n \to \bR$ be any odd function, and say its maximum value $M$ is attained at $x^*=(x^*_1, \ldots, x^*_n)\in\pmone^n$. For $i \in [n]$ define subcube $C_i$ to be the one that sets variables $x_i, \ldots, x_n$ to $x^*_i,\ldots, x^*_n$ respectively and define $C_{n+1}:=\pmone^n$. Note that the set of variables set in each of these subcubes is the complement of a prefix, and hence each of these subcubes is a valid subcube to consider in the denominator of the expression for $\val_\cQ$. Since $f$ is odd and hence balanced, $\E_{x \in C_{n+1}}[f(x)]=0$. We also have that $\E_{x \in C_{1}}[f(x)]=f(x^*)=M$. Thus there exists an $i\in [n]$ such that $\E_{x \in C_{i+1}}[f(x)]-\E_{x \in C_{i}}[f(x)] \leq -M/n$. We will show that $\Var(f_{C_{i+1}})$ is large.

Note that $C_{i}=C_{i+1} \cap \{x: x_i=x^*_i\}$. Let $C'_i:=C_{i+1} \cap \{x: x_i=-x^*_i\}$. Thus $\E_{x \in C_{i+1}}[f(x)] = \frac{1}{2}\rbra{\E_{x \in C'_{i}}[f(x)]+\E_{x \in C_{i}}[f(x)]}$. Hence, 
\begin{align*}
-\frac{M}{n} &\geq \E_{x \in C_{i+1}}[f(x)]-\E_{x \in C_{i}}[f(x)] \\
&=\frac{1}{2}\cdot \rbra{\E_{x \in C'_{i}}[f(x)]-\E_{x \in C_{i}}[f(x)]}\\
&\geq -|\E_{x \in C_{i+1}}[x_i\cdot f(x)]| \\
&=-|\widehat{f_{C_{i+1}}}(\cbra{i})|,
\end{align*}  
implying $|\widehat{f_{C_{i+1}}}(\cbra{i})| \geq M/n$. This in turn implies via Lemma~\ref{lem:fouriervariance} that $\Var(f_{C_{i+1}}) \geq |\widehat{f_{C_{i+1}}}(\cbra{i})|^2 \geq M^2/n^2$. We thus have that $\val_\cQ \leq \frac{M}{\sqrt{M^2/n^2}}=n$.

\item For the lower bound, we define a function such that every valid subcube has small variance. Define $g:\pmone^{n-1} \to \bR$ to be the decision list $((x_{n-1}, 1), (x_{n-2}, 2),$ $\ldots, (x_2, n-2), (x_1, n-1), n)$. Then  define $f:\pmone^n \to \bR$ as follows:
\[
f(x_1, \ldots, x_{n-1}, y) = \begin{cases}
    g(x_1, \ldots, x_{n-1}), & \text{if } y=1,\\
    -g(-x_1, \ldots, -x_{n-1}), & \text{if } y=-1.
\end{cases}
\]
Note that $f$ is odd by construction.
Towards bounding $\val_Q$ from below, we first note that $\|f\|_\infty=n$. Next, we bound the variance of $f$ in each subcube that fixes every variable in the complement of a prefix. We show that the variance of $f$ in each such subcube is $O(1)$. This will show that $\val_\cQ \geq \frac{n}{\sqrt{O(1)}}=\Omega(n)$.

Let $S \in \cQ$ and $b \in \pmone^{\bar{S}}$. We split our analysis into three cases.
\begin{itemize}
\item \textbf{Case 1:} $S=[1,n]$. Since $f$ is odd and hence balanced, we have that
\begin{align*}
\Var(f)&=\E_{x_1, \ldots, x_{n-1}, y}[f(x_1, \ldots, x_{n-1}, y)^2]\\
&=\frac{1}{2}\cdot\E_{x_1, \ldots, x_{n-1}}[f(x_1, \ldots, x_{n-1}, 1)^2]+\frac{1}{2}\cdot\E_{x_1, \ldots, x_{n-1}}[f(x_1, \ldots, x_{n-1}, -1)^2] \\
&=\frac{1}{2}\cdot\E_{x_1, \ldots, x_{n-1}}[g(x_1, \ldots, x_{n-1} )^2]+\frac{1}{2}\cdot\E_{x_1, \ldots, x_{n-1}}[(-g(-x_1, \ldots, -x_{n-1} ))^2] \\
&=\E_{x_1, \ldots, x_{n-1}}[g(x_1, \ldots, x_{n-1} )^2] \\
&=\sum_{i=1}^{n-1}\frac{1}{2^i}\cdot i^2 + \frac{1}{2^{n-1}}\cdot n^2 = \sum_{i=1}^n \frac{i^2}{2^i} + \frac{n^2}{2^n} \leq \sum_{i=1}^{\infty} \frac{i^2}{2^i} + \max_{n\geq1} \frac{n^2}{2^n} < 6+2 = 8.
\end{align*}
\item \textbf{Case 2:} $S=[1,n-1]$. In this case, the restricted function is either $g(x_1, \ldots, x_{n-1})$ or $-g(-x_1, \ldots, -x_{n-1})$ depending on the setting of $y$. In either case, the variance of the restricted function is bounded above by $\E_{x_1, \ldots, x_{n-1}}[g(x_1, \ldots, x_{n-1} )^2]$ which is at most $8$ by the calculations in the previous case.
\item \textbf{Case 3:} $S=[1,i]$ for some $i \in \{1, \ldots, n-2\}$ and $b=(b_{i+1}, \ldots, b_{n-1}, b_0) \in \pmone^{n-i}$. Thus the subcube is obtained by setting $x_j=b_j$ for $j \in \cbra{i+1,\ldots,n-1}$ and $y=b_0$.
\begin{itemize}
\item \textbf{Case 3a:} $b_0=1$. In this case, if $b_j=-1$ for some $j \neq 0$, then $f_{S|b}$ is a constant function, and hence has variance $0$. So, we assume that $b_j=1$ for all $j \in \cbra{i+1, \ldots, n-1}$. Then $f_{S|b}(x_1, \ldots, x_i)$ is the function given by the decision list $((x_i, n-i), (x_{i-1}, n-i+1), \ldots, (x_2, n-2), (x_1, n-1), n)$. 

In order to bound $\Var(f_{S|b})$ we use the following formulation of variance which is easy to verify from the definition of variance.
\[
\Var(f_{S|b})=\frac{1}{2}\cdot\E_{x,y}[(f_{S|b}(x)-f_{S|b}(y))^2].
\]
Note that for $|f_{S|b}(x)-f_{S|b}(y)|$ to be at least $k$, at least one of $f_{S|b}(x)$ and $f_{S|b}(y)$ must be at least $n-i+k$. For $f_{\bar{S}|b}(x)$ ($f_{\bar{S}|b}(y)$, resp.) to be at least $n-i+k$, all $x_j$ (all $y_j$, resp.) for $j \in \cbra{i, i-1, \ldots, i-k-1}$ must equal $1$, an event whose probability is $1/2^k$. We thus have that
\begin{align*}
\Var(f_{S|b})&=\frac{1}{2}\cdot\E_{x,y}[(f_{S|b}(x)-f_{S|b}(y))^2] = \frac12 \sum_{k=0}^n k^2 \mathbb{P}_{x,y}\left[|f_{S|b}(x) - f_{S|b}(y)| = k\right] \\
&\leq\frac{1}{2}\cdot\sum_{k=1}^{i-1} 2\cdot\frac{1}{2^k}\cdot k^2 < 6.
\end{align*}
\item \textbf{Case 3b:} $b_0=-1$. In this case, if $b_j=1$ for some $j \neq 0$, then $f_{S|b}$ is a constant function, and hence has variance $0$. So, we assume that $b_j=-1$ for all $j \in \cbra{i+1, \ldots, n-1}$. Then $f_{S|b}(x_1, \ldots, x_{i-1})$ is the function given by the decision list $((-x_i, -(n-i)), (-x_{i-1}, -(n-i+1)), \ldots, (-x_2, -(n-2)), (-x_1, -(n-1)), -n)$. By a similar analysis as in the previous case, it follows that $\Var(f_{S|b}) < 6$.
\end{itemize}
\end{itemize}
\end{itemize}
The theorem now follows from the definition of $\val_\cQ$ in \Cref{eq:valq}.
\end{proof}

We remark that Theorem~\ref{thm:valprefix} and Theorem~\ref{thm:main} imply a lower bound of $\Omega(n)$ for computing parity using prefix queries. We have not tried to optimize the constants in \Cref{thm:valprefix}. In fact, we show a tighter bound of approximately $n/\sqrt{2}$ in the appendix using the adversary bound directly (see \Cref{lem:prefix_lower_bound_for_dict}) for the function that computes the last input bit $x_n$ (which can only be easier that computing the parity of all the bits).

\subsection{Only Full Set}
Finally, we consider the setting where $\cQ = \cbra{[n]}$. That is, the only allowed substrings must contain all indices. The task of learning $x \in \pmone^n$ is equivalent to search over a space of $2^n$ elements, where one element is promised to be marked.\footnote{Note that the query set $\cQ$ is not downward closed (see \Cref{lem:downwardclosed}) and therefore $\val_\cQ$ only gives us a lower bound on $\Q^{\cQ}(\rec)$.}
Our bound of $\val_\cQ = 2^{(n-1)/2}$ in this case recovers the now well-known lower bound of $\Omega(\sqrt{N})$ of searching for a marked item among $N$ elements~\cite{BBHT98}. The upper bound of $O(\sqrt{N})$ for search follows from \cite{Gro96}.

\begin{theorem}\label{thm:valfull}
    Let $n$ be a positive integer and let $\cQ = \cbra{[n]}$. Then,
    \[
        \val_\cQ = 2^{(n-1)/2}.
    \]
\end{theorem}
\begin{proof}
We first show the lower bound, and then the upper bound. Since $\cQ = \cbra{[n]}$, the denominator of the expression for $\val_\cQ$ (\Cref{eq:valq}) is the standard deviation of $f$ on the whole cube.
\begin{itemize}
    \item For the lower bound, consider $f : \pmone^n \to \bR$ defined by $f(1^n) = 1, f(-1^n) = -1$, and $f(x) = 0$ for all other $x \in \pmone^n$. We have the standard deviation of $f$ on the whole cube equal to $\sqrt{\Var(f)} = \sqrt{\E[f^2] - \E[f]^2} = \sqrt{\E[f^2]} = \sqrt{\frac{2}{2^n}} = 2^{-(n-1)/2}$, and hence $f$ witnesses $\val_\cQ \geq 2^{(n-1)/2}$.
    \item For the upper bound let $f : \pmone^n \to \bR$ be such that $f$ is odd, which in particular means $\E[f] = 0$. The standard deviation over the whole cube equals $\sqrt{\Var(f)} = \sqrt{\E[f^2] - \E[f]^2} = \sqrt{\E[f^2]} \geq \sqrt{\frac{2}{2^n}\max|f(x)|^2}$. Here the last inequality holds because $f$ is odd: if $x^* = \argmax_x |f(x)|$, then since $f(x) = -f(-x)$, we also have $-x^* = \argmax_x |f(x)|$. Thus, for all odd functions $f$ we have $\val_\cQ \leq \frac{\max |f(x)|}{\sqrt{\frac{2}{2^n}\max|f(x)|^2}} = 2^{(n-1)/2}$.\qedhere
\end{itemize}
\end{proof}

\section{Conclusions}\label{sec:conc}
This work introduces a general framework for analyzing query algorithms for learning input strings in models that extend the search with wildcards setting~\cite{AM14}. As applications of our framework we are able to recover existing upper bounds~\cite{AM14, Bel15} in the search-with-wildcards setting and also show new bounds in generalized search-with-wildcards settings. Interestingly, our upper bounds in this paper are all shown (without explicitly resorting to SDP duality) using the \emph{primal} version of the negative-weight adversary bound~\cite{HLS07}, which is a maximization SDP, typically used to show \emph{lower bounds}. To the best of our knowledge, ours is the first work to use the primal adversary bound to establish novel upper bounds. Even though some of our results were already known from previous works~\cite{AM14, Bel15, CILNTTY12, BV23} through various differing techniques, we are able to recover the same bounds through a unifying framework.
Our work makes progress towards answering an open question of~\cite{BBGK18}, who asked if one could devise algorithms for learning input strings or computing some function of the input string, assuming non-standard query access to the input. It would be interesting to see if our results can be used to show any classical lower bounds using the framework developed in~\cite{BBGK18}.

An interesting related direction is that of the AND/OR-query complexity of computing Boolean functions rather than learning an input string. It is not hard to show that the query model in the search-with-wildcards setting is equivalent to allowing unit-cost query access to arbitrary ORs and arbitrary ANDs of inputs. It was recently shown~\cite{CDMRS23} that randomness does not give a superpolynomial advantage over determinism for computing total Boolean functions in this model. It would be interesting to see if ideas from our work could be generalized to show that quantumness also does not provide a superpolynomial advantage over randomness/determinism. Our framework relies on the optimizing matrix for the primal adversary bound being the communication matrix of an XOR function, which crucially uses symmetries of Parity (the function being computed), so it is not clear how to generalize our framework in this direction easily.

\section*{Acknowledgments}
The authors gratefully acknowledge OpenAI’s ChatGPT (GPT-5), which assisted in developing several crucial ideas presented in this work. The authors bear full responsibility for any mistakes in the work. The authors also thank Nadezhda Voronova for pointing out that one of our lower bounds (when the query set is contiguous blocks) was already known, and for pointing us to the relevant references.

A.C.\ is supported by a
Simons-CIQC postdoctoral fellowship through NSF QLCI Grant No.\ 2016245. N.S.M.~gratefully acknowledges support from the London Mathematical Society through a Scheme 7 grant, which enabled preliminary work that sparked this project. S.P.\ and N.R.\ acknowledge the
support from the Dutch Ministry of Education, Culture, and Science through Gravitation project ``Challenges
in Cyber Security - 024.006.037'' for this work.

\bibliography{bibo}

\appendix

\section{Proof of \Cref{lem:equivalence_of_formulations}}\label{sec:equivalence_of_formulations_proof}

We first observe that for any Boolean function $F$ on $D \subseteq \{-1,1\}^n$, the optimal adversary matrix $\Gamma$ for $\ADV^{\cQ,\pm}(F)$ attains its norm without needing absolute values, i.e., its spectral norm equals its largest eigenvalue.

\begin{lemma}\label{lem:eigenvalue_sign}
    Let $D\subseteq\{-1,1\}^n$ and $F:D\rightarrow\{-1,1\}$ be a Boolean function. Let $\Gamma$ be an optimal solution to $\ADV^{\cQ,\pm}(F)$. Then $\|\Gamma\| = \lambda_{\max}(\Gamma)$.
\end{lemma}

\begin{proof}

    To prove the lemma, it is sufficient to show that if $\lambda$ is an eigenvalue of $\Gamma$, then so is $-\lambda$. We first group the elements in $D$ as $F^{-1}(1)$ and $F^{-1}(-1)$. Grouping the indices of $\Gamma$, it takes the following form.

    \begin{equation}\label{eqn:boolean_feasible_solution}
        \Gamma = \begin{bmatrix}
            0 & A \\
            A^T & 0
        \end{bmatrix}.
    \end{equation}

    Let the rank of the matrix $A$ be $r$. Then, the rank of the matrix $\Gamma$ is $2r$. Furthermore, let $A = \sum_{j=1}^r s_j\ket{u_j}\bra{v_j}$ be the singular value decomposition of $A$. Now, for every $j \in [r]$, we find
    \[\Gamma(\ket{v_j} \oplus \pm\ket{u_j}) = (\pm A^T\ket{u_j}) \oplus A\ket{v_j} = s_j(\pm\ket{v_j} \oplus \ket{u_j}) = \pm s_j(\ket{v_j} \oplus \pm\ket{u_j}),\]
    and so we have found all $2r$ eigenvalues of $\Gamma$. It follows that the spectrum is $\{s_j, - s_j : j \in [r]\}$, which is symmetric around $0$. It follows that $\norm{\Gamma} = \lambda_{\max}(\Gamma)$.
\end{proof}

We require the following fact that is easy to verify.
\begin{fact}\label{fact:commute}
    For a diagonal matrix $A$, and matrices $B, C$, we have $A(B\circ C)A = (ABA)\circ C$.
\end{fact}

For the sake of clarity, we restate \Cref{lem:equivalence_of_formulations}.

\begin{repeattheorem}{\Cref{lem:equivalence_of_formulations}}
    Let $D\subseteq\{-1,1\}^n$ and $F:D\rightarrow\{-1,1\}$ be a Boolean function. Then $\ADV^{\cR,\pm}(F) = \overline{\ADV}^{\cR,\pm}(F)$. Moreover, if $(\beta,\Gamma)$ is an optimal solution for $\overline{\ADV}^{\cR,\pm}(F)$, then $\Gamma'$, defined as\footnote{Here, we use the convention that $0/0 = 0$.}
    \[\Gamma'[x,y] = \frac{\Gamma[x,y]}{\sqrt{\beta[x]\beta[y]}},\]
    is an optimal solution for $\ADV^{\cR,\pm}(F)$.
\end{repeattheorem}

\begin{proof}
    We split our proof into two parts. First, we show how to convert a feasible solution $\Gamma^\prime$ of $\ADV^{\cR,\pm}(F)$ into a feasible solution $(\beta,\Gamma)$ of $\overline{\ADV}^{\cR,\pm}(F)$ with the same objective value. Next, we show that if $(\beta,\Gamma)$ is an optimal solution for $\overline{\ADV}^{\cR,\pm}(F)$, then $\Gamma^\prime$ is an optimal solution for $\ADV^{\cR,\pm}(F)$.

    \paragraph{Standard formulation to alternate formulation} Let $M$ be an optimal solution to $\ADV^{\cR,\pm}(F)$. Then there exists a unit vector $\ket{v}$ such that $M\ket{v} = \lambda_{\max}(M)\ket{v}$, and so
    
    \begin{align}\label{eq:normmaxeval}
        \|M\| & = \left|\lambda_{\max} (M)\right|, \\
        \nonumber& = \lambda_{\max}(M),\tag{by \Cref{lem:eigenvalue_sign}}\\
        \nonumber& = \bra{v}M\ket{v}\tag{$M\ket{v} = \lambda_{\max}(M)\ket{v}$},
    \end{align}
    If an entry of $\ket{v}$ has a negative sign, we can switch this to a positive sign and instead multiply the corresponding row and column of $M$ with a negative sign. Define a diagonal matrix $A\in\{-1,0,1\}^{D\times D}$ as follows.

    \begin{equation*}
        A[x,x] = \begin{cases}
            -1, & v[x] < 0, \\
            1, & v[x] \ge 0.
        \end{cases}
    \end{equation*}
    
    We write $\ket{v_{\mathsf{abs}}}$ for the vector containing the entry-wise absolute value of $\ket{v}$. Then $A\ket{v_{\mathsf{abs}}} = \ket{v}$, and   

    \begin{equation}\label{eq:vabs}
        \bra{v}M\ket{v} = \bra{v_{\mathsf{abs}}}A M A\ket{v_{\mathsf{abs}}} = \sum_{x,y}|v[x]|A[x,x]M[x,y]A[y,y]|v[y]|.
    \end{equation}
    Since $M$ is an optimal (and hence feasible) solution to $\ADV^{\cR,\pm}(F)$, $F(x) = F(y) \implies M[x, y] = 0$, which implies, since $AMA$ is simply flipping the signs of certain rows and columns of $M$, that $F(x) = F(y) \implies AMA[x,y] = 0$, which means $AMA$ is a feasible solution to $\ADV^{\cR,\pm}(F)$ as well. By \Cref{eq:vabs}, this new solution is also optimal, and hence we can assume without loss of generality that all entries of $\ket{v}$ are non-negative. For the rest of the proof we assume that $M$ is an optimal solution to $\ADV^{\cR,\pm}(F)$ whose principal eigenvector $v$ has only non-negative entries.
    
    Now, we define a vector $\beta$ such that $\beta[x] = v[x]^2$. Then, \Cref{eq:vabs} implies
    
    \begin{equation}\label{eqn:objvalusingbeta}
        \bra{v}M\ket{v} = \sum_{x,y\in D}\sqrt{\beta[x]}M[x,y]\sqrt{\beta[y]}.
    \end{equation}

    Define the quantity $s = \sum_{x\in F^{-1}(1)}\beta[x]$. By the definition of $M$, the contribution to the sum $\sum_{x,y\in D}\sqrt{\beta[x]}M[x,y]\sqrt{\beta[y]}$ from $(x,y)$ with $F(x) = F(y)$ is $0$. This means every non-zero summand must correspond to $x \in F^{-1}(1), y \in F^{-1}(-1)$ (or vice versa). Since the objective value considered is non-zero, we must have $s > 0$. Note that since $\|v\| = 1$, this means $\sum_x \beta[x] = 1$, which implies $s \leq 1$ (a similar argument to the above shows $s < 1$).
    Consider vector $\beta^\prime\in\mathbb{R}^D$ whose entries are defined as follows.

    \begin{equation*}
        \beta^\prime[x] = \begin{cases}
            \frac{1}{2s}\beta[x], & x\in F^{-1}(1)\\
            \frac{1}{2(1-s)}\beta[x], & x\in F^{-1}(-1)
        \end{cases}
    \end{equation*}

    We have
    \begin{align}\label{eq:sqrtbetaprimebeta}
        & \sum_{x,y\in D}\sqrt{\beta^\prime[x]}M[x,y]\sqrt{\beta^\prime[y]},\\
        \nonumber = & \sum_{x,y\in D}\frac{1}{2\sqrt{s(1-s)}}\sqrt{\beta[x]}M[x,y]\sqrt{\beta[y]},\tag{$F(x) = F(y) \implies M[x,y] = 0$}\\
        \nonumber\ge & \sum_{x,y\in D}\sqrt{\beta[x]}M[x,y]\sqrt{\beta[y]}.\tag{$s\in (0,1)$}
    \end{align}

    This choice of $\beta^\prime$ guarantees that
    \begin{equation}\label{eqn:beta_balance}
        \sum_{x \in F^{-1}(1)} \beta^\prime[x] = \sum_{x \in F^{-1}(-1)} \beta^\prime[x] = \frac{1}{2}.
    \end{equation}

Define the diagonal matrix $B\in\mathbb{R}^{D\times D}$ such that $B[x,x] = \sqrt{\beta^\prime[x]}$. Next, we define a new matrix $\Gamma = B M B$.

    \begin{equation}\label{eq:sumgammabeta}
        \sum_{x,y\in D}\Gamma[x,y] = \sum_{x,y\in D}\sqrt{\beta^\prime[x]}M[x,y]\sqrt{\beta^\prime[y]}.
    \end{equation}

    Observe that both $M$ and $\Delta_q$ are symmetric matrices. Therefore, $M\circ\Delta_q$ is also a symmetric matrix. Due to this, $\|M\circ\Delta_q\| = |\lambda_{\max}(M\circ\Delta_q)|$. Then
    \begin{equation*}
        \|M\circ\Delta_q\| \le 1 \implies |\lambda_{\max}(M\circ\Delta_q)| \le 1,
    \end{equation*}
    which implies all eigenvalues of $M \circ \Delta_q$ are in $[-1,1]$.
    Note that any vector in $\mathbb{R}^D$ is an eigenvector of $I_D$ with eigenvalue $1$. Then $I_D - M\circ\Delta_q$ has the same eigenvectors as $M\circ\Delta_q$ with eigenvalues of the form $1-\lambda$ where $\lambda\in[-1,1]$. Therefore, all eigenvalues of $I_D - M\circ\Delta_q$ are non-negative and hence $I_D - M\circ\Delta_q \succeq 0$. This implies, since $B$ is diagonal with all non-negative entries,
    
    \begin{align}\label{eqn:feasiblePSD}
        B(I_D - M\circ\Delta_q)B & \succeq 0,\\
        \nonumber BB - B M B\circ\Delta_q & \succeq 0,\tag*{by \Cref{fact:commute}}\\
        \nonumber \mathsf{diag}(\beta') - \Gamma\circ\Delta_q & \succeq 0.
    \end{align}

    Note that $(\beta', \Gamma)$ has the following properties.
    \begin{itemize}
        \item $\Gamma = BMB$ is symmetric since $B$ is diagonal and $M$ is symmetric, $\beta'$ is entry-wise non-negative by construction.
        \item $\Gamma[x,y] = 0$ if $F(x) = F(y)$ since $\Gamma[x,y] = \sqrt{\beta'[x]}M[x,y]\sqrt{\beta'[y]}$ and $M[x,y] = 0$ if $F(x) = F(y)$.
        \item $\mathsf{diag}(\beta) - \Gamma\circ\Delta_q \succeq 0$ for all $q\in \cR$ due to \Cref{eqn:feasiblePSD}
        \item $\sum_{x\in F^{-1}(1)}\beta'[x] = \sum_{x\in F^{-1}(-1)}\beta'[x] = \frac{1}{2}$ due to \Cref{eqn:beta_balance}
    \end{itemize}
    Thus, $(\beta', \Gamma)$ is a feasible solution to $\overline{\ADV}^{\cR, \pm}(F)$, and the objective value attained is 
    \begin{align*}
        \sum_{x, y \in D}\Gamma[x,y] & = \sum_{x,y\in D}\sqrt{\beta^\prime[x]}M[x,y]\sqrt{\beta^\prime[y]} \tag{by \Cref{eq:sumgammabeta}}\\
        & \geq \sum_{x,y\in D}\sqrt{\beta[x]}M[x,y]\sqrt{\beta[y]} \tag{by \Cref{eq:sqrtbetaprimebeta}}\\
        & = \bra{v}M\ket{v} \tag{by \Cref{eqn:objvalusingbeta}}\\
        & = \|M\| \tag{by \Cref{eq:normmaxeval}}.
    \end{align*}
    This shows that $\overline{\ADV}^{\cR,\pm}(F) \ge \ADV^{\cR,\pm}(F)$.

    \paragraph{Alternate formulation to standard formulation} In the other direction, let $(\beta,\Gamma)$ be an optimal solution to $\overline{\ADV}^{\cR, \pm}(F)$. Define $\Gamma^\prime = B\Gamma B$ where the matrix $B$ is a diagonal matrix defined as follows.

    \begin{equation}
        B[x,x] = \begin{cases}
            \frac{1}{\sqrt{\beta[x]}}, & \beta[x]\neq 0\\
            0, & \text{otherwise}
        \end{cases}
    \end{equation}

    Suppose that $\Gamma[x,y] > 0$. Then, $F(x) \neq F(y)$, which in particular means that $x \neq y$. Let $q$ be a query that distinguishes between $x$ and $y$. Since $\mathsf{diag}(\beta) - \Gamma\circ\Delta_q \succeq 0$, so are all its submatrices. In particular, the $2 \times 2$ submatrix of indexed by $x$ and $y$ is also PSD, and it looks as follows;
    \[\begin{bmatrix}
        \beta[x] & -\Gamma[x,y] \\
        -\Gamma[x,y] & \beta[y]
    \end{bmatrix} \succeq 0.\]
    For a $2 \times 2$ matrix to be PSD, the determinant must be non-negative, which implies that $\beta[x]\beta[y] - \Gamma[x,y]^2 \geq 0$. Since $\Gamma[x,y] \neq 0$ by assumption, we find that $\beta[x]\beta[y] \geq \Gamma[x,y]^2 > 0$, and so we find that $\beta[x] > 0$ and $\beta[y] > 0$. By the contrapositive, if $\beta[x] = 0$, then the entire row and column in $\Gamma$ labeled by $x$ is $0$. As such, we can rewrite the objective value in $\overline{\ADV}^{\cR,\pm}(F)$ as
    \[\overline{\ADV}^{\cR,\pm}(F) = \sum_{x,y \in D}\Gamma[x,y] = \sum_{\substack{x,y  \in D \\ \beta[x] > 0 \land \beta[y] > 0}}\Gamma[x,y].\]
    
    Now, since $(\beta,\Gamma)$ is a feasible solution, $\mathsf{diag}(\beta) - \Gamma\circ\Delta_q \succeq 0$. Then
    \begin{align}
        B(\mathsf{diag}(\beta) - \Gamma\circ\Delta_q)B & \succeq 0,\\
        \nonumber B \mathsf{diag}(\beta) B - B\Gamma B\circ\Delta_q & \succeq 0,\tag*{by \Cref{fact:commute}}\\
        \nonumber I_D - \Gamma^\prime\circ\Delta_q & \succeq 0.
    \end{align}
    This implies that eigenvalues of $I_D - \Gamma^\prime\circ\Delta_q$ are non-negative. Since $I_D - \Gamma^\prime\circ\Delta_q$ and $\Gamma^\prime\circ\Delta_q$ have the same eigenvectors, eigenvalues of $\Gamma^\prime\circ\Delta_q$ are at most $1$. Furthermore, $\Gamma'$ and $\Gamma'\circ\Delta_q$ are of the form 
    \[
    \begin{bmatrix}
        0 & A\\
        A^T & 0
    \end{bmatrix}.
    \]
    This observation allows us to use the same argument as in the proof of \Cref{lem:eigenvalue_sign}, implying that for all eigenvalues $\lambda$, $-\lambda$ is also an eigenvalue. Combining these two observations, we get $\|\Gamma^\prime\circ\Delta_q\|\le 1$.

    Observe that $\Gamma'$ has the following properties.
    \begin{itemize}
        \item $\Gamma' = B\Gamma B$ is symmetric since $B$ is diagonal and $\Gamma$ is symmetric.
        \item Since 0-entries in $\Gamma$ continue to be 0-entries in $\Gamma'$, this means $F(x) = F(y) \implies \Gamma'(x,y) = 0$. 
        \item $\|\Gamma'\circ\Delta_q\| \leq 1$ for all $q\in\cR$ as shown above. 
    \end{itemize}

    Let the vector $\ket{v}$ be defined as $v[x] = \sqrt{\beta[x]}$. Then
    \begin{align*}
        \ADV^{\cR,\pm}(F) &\geq \norm{\Gamma'} \geq \bra{v}\Gamma'\ket{v} = \sum_{x,y}\sqrt{\beta[x]}\Gamma'[x,y]\sqrt{\beta[y]} \\
        &= \sum_{\substack{x,y \\ \beta[x] > 0 \land \beta[y] > 0}}\sqrt{\beta[x]}\frac{1}{\sqrt{\beta[x]}}\Gamma[x,y]\frac{1}{\sqrt{\beta[y]}}\sqrt{\beta[y]} = \sum_{\substack{x,y \\ \beta[x] > 0 \land \beta[y] > 0}}\Gamma[x,y] = \overline{\ADV}^{\cR,\pm}(F).
    \end{align*}

    Moreover, from an optimal solution $(\beta,\Gamma)$ for $\overline{\ADV}^{\cR,\pm}(F)$, we obtained a solution $\Gamma'$ to $\ADV^{\cR,\pm}(F)$ defined by $\Gamma'[x,y] = \frac{\Gamma[x,y]}{\sqrt{\beta[x]\beta[y]}}$ with at least as large objective value (and hence equal, by the first part of this proof). This concludes the proof.\qedhere
\end{proof}

\section{Prefix-Query Complexity of Dictator Function}\label{sec:prefix_dictator}

In this section we study the query complexity of the Dictator function $\DictatorFunction_n : \pmone^n \to \pmone$ defined by $\DictatorFunction_n(x) = x_n$ in the substring query model with $\cQ = \{ [1,i] \mid i\in[n] \}$ as the set of all prefixes.
\begin{lemma}\label{lem:prefix_lower_bound_for_dict}
    Let $n$ be a positive integer. Let $\cal{Q} = \{ [1,i] \mid i\in[n] \}$. Then,
    \begin{equation*}
        \ADV^{\cal{Q},\pm}(\mathsf{Dict}_n) \ge 2+\frac{n-3}{\sqrt{2}}.
    \end{equation*}
\end{lemma}

\begin{proof}
    We prove a lower bound for $\ADV^{\cQ, \pm}(\DictatorFunction_n)$ by exhibiting a feasible solution, denoted by $\Gamma$, to the standard formulation of the primal adversary bound for the $\DictatorFunction_n$ function and by analyzing the quantity
    \begin{equation}
    \label{eqn:quantity}
        \frac{\norm{\Gamma}}{\max\limits_{\substack{S \in \cQ,\\ b \in \pmone^S}} \norm{\Gamma \circ \Delta_{S,b}}}.
    \end{equation}

    \paragraph{Defining matrix $\Gamma$} Let $\IndexFunction: \pmone^n \rightarrow[n]\cup \{0\}$ be a function defined as $\IndexFunction(x) = 0$ if $x_n = 1$, and $\IndexFunction(x) = \min\cbra{i \in [n] : x_i = -1}$ otherwise.

    Let $f:\{-1,1\}^n \rightarrow \mathbb{R}$ be a function defined as follows.
    \begin{equation*}
        f(x) = \begin{cases}
            0, & x_n = 1\\
            \frac{1}{2}, & \IndexFunction(x) \in \{n-1,n\}\\
            \frac{1}{2^{d-1}\sqrt{2}}, & \IndexFunction(x) = n-d \quad \forall d\in[2,n-2]\\
            \frac{1}{2^{n-2}}, &\IndexFunction(x) = 1.
        \end{cases}
    \end{equation*}

    Let $\Gamma \in \mathbb{R}^{2^n \times 2^n}$ be the matrix defined as $\Gamma = M_{f \circ \XOR}$.

    First observe that this is a valid matrix to consider for the adversary bound: it is symmetric, and $\DictatorFunction_n(x) = \DictatorFunction_n(y) \implies (x \oplus y)_n = 1 \implies f(x \oplus y) = 0$.
    To compute the quantity in \Cref{eqn:quantity}, we first reorder the rows and similarly the columns of $\Gamma$ in a way that $\Gamma$ can be partitioned into blocks containing all strings with $x_n = 1$ appear first, followed by all strings with $x_n=-1$. Clearly, with such an ordering the two diagonal blocks will contain all zeroes. Furthermore, the remaining two off-diagonal blocks will be transpose of each other because $\Gamma$ is a symmetric matrix by construction. Consequently, $\Gamma$ is of the form 

    \begin{equation}
    \label{eq:BlockStructureGamma}
        \Gamma = \begin{bmatrix}
            0 & A\\
            A^T& 0\\
        \end{bmatrix}.
    \end{equation}

    Further ordering is simply lexicographic in the reverse of the strings. Equivalently this ordering is done by the value of the function $\IndexFunction(\cdot)$: the bit-string $x$ appears before the bit-string $y$ if $\IndexFunction(x) > \IndexFunction(y)$, and ties are broken by similarly comparing the values of $\IndexFunction$ on the substrings of $x$ and $y$ obtained by deleting the first $\IndexFunction(x)$ bits.

    \paragraph{Lower bounding numerator of \Cref{eqn:quantity}}
    Let $x \in \pmone^n$ be an arbitrary string. The $x$'th row of $\Gamma$ contains $2^{n-1}$ entries with value $0$ (corresponding to $y$ with $x_n = y_n$), $2^{n-2}$ entries with value $\frac{1}{2^{n-2}}$ (corresponding to all remaining $y$ with $x_1 \neq y_1$), $2^{d-1}$ entries of $\frac{1}{2^{d-1}\sqrt{2}}$ $\forall d\in[2,n-2]$ (corresponding to all $y$ with $x_n \neq y_n$ and $\IndexFunction(x \oplus y) = n-d$) and $2$ entries of values $\frac{1}{2}$ (corresponding to the input strings $(1^{n-2},-1,-1)$ and $(1^{n-1},-1)$). Let $\ket{\bar{1}}$ denote the all-$1$ vector of length $2^n$. Then,
    \begin{equation*}
        \frac{1}{2^n}\bra{\bar{1}}\Gamma\ket{\bar{1}} = \frac{1}{2^n}\left[2^n\left(\frac{1}{2}+\frac{1}{2}+\sum_{d=2}^{n-2}\frac{2^{d-1}}{2^{d-1}\sqrt{2}}+\frac{2^{n-2}}{2^{n-2}}\right)\right] = 2+\frac{n-3}{\sqrt{2}}.
    \end{equation*}
    
    Since $\norm{\Gamma} \ge \frac{1}{2^n}\bra{\bar{1}}\Gamma\ket{\bar{1}}$, we get $\norm{\Gamma} \ge 2+\frac{n-3}{\sqrt{2}}$.\\

    \paragraph{Computing denominator of \Cref{eqn:quantity}}
    Recall that $G$ is the bit-flip group (see \Cref{def:bitflipgroup}). Since $\Gamma$ is a communication matrix of a $\XOR$ function, $\forall \sigma\in G, \Gamma = \sigma\Gamma\sigma^T$. Consider an arbitrary query $q = ([k],b)$ where $b\in\{-1,1\}^k$. Let $\sigma := \sigma_{b1^{n-k}}$ That is, $\sigma$ is the permutation in $G$ such that $\sigma(b1^{n-k}) = 1^n$.
    \begin{align*}
        & \sigma(\Gamma\circ\Delta_q)\sigma^T = (\sigma\Gamma\sigma^T)\circ(\sigma\Delta_q\sigma^T) = \Gamma\circ\Delta_{1^k},\\
        \implies & \norm{\Gamma\circ\Delta_q} = \norm{\sigma(\Gamma\circ\Delta_q)\sigma^T} = \norm{\Gamma\circ\Delta_{1^k}}.
    \end{align*}

    Therefore, it is sufficient to bound the denominator for queries of the form $q=([k],1^k), \forall k\in[n]$.
    
    \paragraph{Case $q=(\{1\}, 1)$}
    The matrix $\Gamma\circ\Delta_q$ is the following block matrix where each block is of size $2^{n-2}$.
    
    \begin{equation}
        \Gamma\circ\Delta_q = \frac{1}{2^{n-2}}\left(\begin{bmatrix}
            0 & 0 & 0 & 1\\
            0 & 0 & 0 & 0\\
            0 & 0 & 0 & 0\\
            1 & 0 & 0 & 0\\
        \end{bmatrix}_{4\times 4} \otimes J_{2^{n-2} \times 2^{n-2}}\right).
    \end{equation}
    
    Clearly, this is a rank $2$ matrix with eigenvectors $(v, \bar{0}, \bar{0}, v)^T$ and $(v, \bar{0}, \bar{0}, -v)^T$ where $\bar{0}$ represents the all zero column vector of length $2^{n-2}$ and $v$ is the all one column vector of length $2^{n-2}$ multiplied with $\frac{1}{2^{n-2}}$. The corresponding eigenvalues are $+1$ and $-1$ respectively.

    \paragraph{Case $q = ([n], 1^n)$}
    The matrix $\Gamma\circ\Delta_q$ is the following.
    
    \begin{align}
        \Gamma\circ\Delta_q & = \frac{1}{2^{n-2}}\begin{bmatrix}
            0 & \cdots & 0 & A\\
            \vdots & \ddots & \vdots & 0_{1\times2^{n-1}}\\
            0 & \cdots & 0 & 0\\
            A^T & 0_{2^{n-1}\times1} & 0 & 0_{2^{n-1}\times2^{n-1}} \\
        \end{bmatrix},
    \end{align}
    where $A$ is the row vector defined as 
    \begin{align*}
        A = (\frac{1}{2}, \frac{1}{2}, \frac{1}{2\sqrt{2}}, \frac{1}{2\sqrt{2}}, \underbrace{\frac{1}{4\sqrt{2}}, \ldots, \frac{1}{4\sqrt{2}}}_{2^2}, \ldots, \underbrace{\frac{1}{2^{n-3}\sqrt{2}}, \ldots, \frac{1}{2^{n-3}\sqrt{2}}}_{2^{n-3}}, \underbrace{\frac{1}{2^{n-2}}, \ldots, \frac{1}{2^{n-2}}}_{2^{n-2}}).
    \end{align*}

    Once again, $\Gamma\circ\Delta_q$ is a rank $2$ matrix. The eigenvectors are the following with the appropriate normalization factors.
    \begin{equation}
        (\pm1, 0, \ldots, 0, \frac{1}{2}, \frac{1}{2}, \frac{1}{2\sqrt{2}}, \frac{1}{2\sqrt{2}}, \underbrace{\frac{1}{4\sqrt{2}}, \ldots, \frac{1}{4\sqrt{2}}}_{2^2}, \ldots, \underbrace{\frac{1}{2^{n-3}\sqrt{2}}, \ldots, \frac{1}{2^{n-3}\sqrt{2}}}_{2^{n-3}}, \underbrace{\frac{1}{2^{n-2}}, \ldots, \frac{1}{2^{n-2}}}_{2^{n-2}}).
    \end{equation}
    The corresponding eigenvalues are $\pm1$ respectively.

    \paragraph{Case $q = ([k], 1^k), k\in[2,n-1]$}
    The corresponding matrix $\Gamma\circ\Delta_q$ is of the form $\begin{bmatrix}
        0 & B\\
        B & 0\\
    \end{bmatrix}$ where B is a $2^{n-1}\times2^{n-1}$ symmetric matrix defined below. The eigenvectors of $\Gamma\circ\Delta_q$ are $(v, v)^T$ and $(v, -v)^T$ where $v$ is any eigenvector of $B$ and the eigenvalues are $\pm\lambda$ where $\lambda$ is the eigenvalue of $B$ corresponding to eigenvector $v$.

    \begin{align}
        B & = \begin{bmatrix}
            0 & A_0 & A_1 & \ldots & A_{k-2} & A_{k-1}\\
            A_0^T & & & & & \\
            A_1^T & & & & & \\
            \vdots & & & {\textrm{\huge 0}} & & \\
            A_{k-2}^T & & & & & \\
            A_{k-1}^T & & & & & \\
        \end{bmatrix},\\
        \nonumber A_i & = \frac{1}{2^{n-k-1+i}\sqrt{2}}[1]_{(2^{n-k-1})\times (2^{n-k-1+i})},\tag{$\forall i\in\{0,1,\ldots,k-2\}$}\\
        \nonumber A_{k-1} & = \frac{1}{2^{n-2}}[1]_{(2^{n-k-1})\times (2^{n-2})}.
    \end{align}

    The matrix $B$ is clearly a rank $2$ matrix since each $A_i$ is a constant matrix, and each $A_i$ has the same height. The two eigenvectors of the matrix are the following with the appropriate normalization factor. The corresponding eigenvalues are $\pm1$. Therefore, the eigenvalues of $\Gamma\circ\Delta_q$ are also $\pm1$.

    \begin{equation*}
        (\underbrace{\pm1,\ldots,\pm1}_{2^{n-k-1}},\underbrace{\frac{1}{\sqrt{2}},\ldots,\frac{1}{\sqrt{2}}}_{2^{n-k-1}},\underbrace{\frac{1}{2\sqrt{2}},\ldots,\frac{1}{2\sqrt{2}}}_{2^{n-k}},\ldots,\underbrace{\frac{1}{2^{k-2}\sqrt{2}},\ldots,\frac{1}{2^{k-2}\sqrt{2}}}_{2^{n-3}},\underbrace{\frac{1}{2^{k-1}\sqrt{2}},\ldots,\frac{1}{2^{k-1}\sqrt{2}}}_{2^{n-2}})^T.
    \end{equation*}

    For every prefix query $q$, the quantity $\norm{\Gamma\circ\Delta_q}$ is always $1$. Therefore, $\max_q\norm{\Gamma\circ\Delta_q} = 1$.

    \paragraph{Adversary bound} Combining the lower bound for the numerator with the denominator value, we get the following.

    \begin{equation*}
      \ADV^{\cQ, \pm}(\mathsf{Dict_n}) = \frac{\norm{\Gamma}}{\max\limits_{\substack{S \in \cQ,\\ b \in \pmone^S}} \norm{\Gamma \circ \Delta_{S,b}}} \ge  2+\frac{n-3}{\sqrt{2}}.\qedhere
    \end{equation*}
\end{proof}

\Cref{thm:adv_equals_quantum_queries} and the naive upper bound of query prefixes of increasing size, learning $1$ bit in each step, immediately yields the following corollary.

\begin{corollary} Let $n$ be a positive integer. Let $\cQ=\{[1,i] \mid i \in [n]\}$. Then,
    \begin{equation*}
        \Q^{\cQ}(\DictatorFunction_n)=\Theta(n).
    \end{equation*}
\end{corollary}

\begin{remark}
    We observe that when $\cQ=\{[1,i] \mid i \in [n]\}$, computing parity is at least as hard as computing the last bit, more formally, $\Q^{\cQ}(\oplus) \geq (1/2)\Q^{\cQ}(\DictatorFunction_n)$: 
    Given an algorithm for computing parity, first run this algorithm on the input string $x$ and get output $b$, say. Next run this algorithm on the input string with the last bit deleted (note that prefix queries to this input are valid prefix queries to the original input) and get output $b'$, say. It is easy to see that $x_n = b \oplus b'$.
\end{remark}

\end{document}